\documentclass[english, thm-restate, notab]{lipics-v2021}
\usepackage[utf8]{inputenc}
\usepackage{amsmath}
\usepackage{amsfonts}
\usepackage{amsthm}
\usepackage{makecell}
\usepackage{xspace}
\usepackage{caption}
\usepackage{imakeidx}
\usepackage{thmtools}
\usepackage{makecell}
\usepackage{todonotes}
\usepackage{algorithm}
\usepackage{algpseudocode}
\usepackage{booktabs}
\usepackage{ragged2e}
\usepackage{xspace}
\usepackage{adjustbox}
\usepackage{changepage}
\usepackage{todonotes}
\usepackage{tcolorbox}
\usepackage{graphicx,caption}
\usepackage{cite}
\usepackage{eucal}

\include{tikzsetup}

\author{Jacobus Conradi}{Department of Computer Science, University of Copenhagen, Denmark}{jaco@di.ku.dk}{}{}

\author{Ivor van der Hoog}{IT University of Copenhagen, Denmark}{ivva@itu.dk}{
https://orcid.org/0009-0006-2624-0231}{}

\author{Eva Rotenberg}{IT University of Copenhagen, Denmark}{erot@itu.dk}{
https://orcid.org/0000-0001-5853-7909}{}

\funding{{\it Ivor van der Hoog} and {\it Eva Rotenberg} are grateful to the Carlsberg Foundation for supporting this research via Eva Rotenberg's Young Researcher Fellowship CF21-0302 ``Graph Algorithms with Geometric Applications''. \textit{Jacobus Conradi} is supported by the Carlsberg Foundation, grant CF24-1929.}

\hideLIPIcs
\linenumbers

\begin{document}

\newcommand{\fancy}[1]{\EuScript{#1}}
\newcommand{\dfs}{\Delta\text{-FSD}}

\title{{\fontsize{16pt}{22pt}\selectfont On computing the (exact) Fréchet distance with a frog}}
\titlerunning{On computing the (exact) Fréchet distance with a frog}

\newcommand{\frechet}{Fr\'{e}chet\xspace}
\newcommand{\Cov}{\texttt{\small cov}}
\newcommand{\fd}{\mathcal{D}_F}
\newcommand{\dd}{\mathbb{D}_F}
\newcommand{\ve}{\mathcal{V}_E}
\newcommand{\ive}{\mathcal{I}_{VE}}

\newtheorem{algodef}{Clustering}

\nolinenumbers

\bibliographystyle{plainurl}

\supplement{}
\supplementdetails[subcategory={Source code}, cite={}, swhid={}]{Software}{https://anonymous.4open.science/r/exactFrechetUnleashed-62D1/}

\EventEditors{XX}
\EventNoEds{2}
\EventLongTitle{XXX}
\EventShortTitle{XX}
\EventAcronym{XX}
\EventYear{XX}
\EventDate{XX}
\EventLocation{XX}
\EventLogo{socg-logo.pdf}
\SeriesVolume{X}
\ArticleNo{XX}

\authorrunning{J. Conradi, I. van der Hoog, and E. Rotenberg}
\Copyright{J. Conradi, I. van der Hoog, and E. Rotenberg}


\date{}

\maketitle

\begin{abstract}
The continuous Fréchet distance $\mathcal{D}_F(\pi,\sigma)$ between two polygonal curves $\pi$ and $\sigma$ is classically computed by exploring the \emph{free space diagram} over the two curves. 
Har-Peled, Raichel, and Robson [SoCG'25] recently proposed a radically different approach: they approximate $\mathcal{D}_F(\pi,\sigma)$ by computing paths in a discrete graph that models a joint traversal of~$\pi$ and~$\sigma$, recursively bisecting edges until the discrete distance converges to the continuous one. They implement their ``frog-based'' technique, and claim that it yields substantial practical speedups compared to the state-of-the-art implementations.

In this paper, we revisit this technique as we address three aspects of the paper: (i) it does not compute the exact Fréchet distance, (ii) its recursive bisection introduces the required monotonicity events to realise the Fréchet distance only in the limit, and (iii) it applies a heuristic simplification technique which is conservative. 
Motivated by theoretical interest, we develop new techniques that guarantee exactness, polynomial-time convergence and near-optimal lossless simplifications. We provide an open-source C++ implementation of our variant.

Our primary contribution is an extensive empirical evaluation on a broad, publically available, suite of real-world and synthetic data sets. 
Among the frog-based variants, our exact computation introduces a level of overhead. This is to be expected, and we see an increase in the median runtime. Yet, our new approach is often faster in the worst case, worst ten percent, or even the average runtime due to its worst-case convergence guarantees.  More surprisingly, our analysis shows that the implementation of Bringmann, Künnemann, and Nusser [SoCG'19] 
dominates all frog-based implementations in performance.
These results provide a much-needed nuanced perspective on the capabilities and limitations of frog-based techniques: we showcase its theoretical appeal, but highlight its current practical limitations. 
\end{abstract}
\clearpage 

\section{Introduction}

The continuous Fréchet distance between two polygonal curves $\pi$ and $\sigma$ is classically described using the analogy of a person walking along $\pi$ and a dog walking along $\sigma$, each moving continuously and monotonically. The continuous Fréchet distance $\fd(\pi, \sigma)$ is the infimum over all such traversals of the maximum distance between them. A discrete analogue is defined similarly by imagining two frogs that hop monotonically along the vertices of $\pi$ and $\sigma$, leading to the discrete Fréchet distance $\dd(\pi, \sigma)$.

Let $\pi = (p_1,\ldots,p_n)$ and $\sigma = (s_1,\ldots,s_m)$. From an algorithmic and practical perspective, it is possible to decide whether their continuous or discrete distance is less than some input value $\Delta$ in an efficient manner. For a given value $\Delta$, deciding whether $\dd(\pi, \sigma) \leq \Delta$ reduces to determining whether there exists a monotone path in an $n \times m$ vertex-weighted grid where the weight of $(p_i, s_j)$ is $d(p_i, s_j)$. Removing all vertices with weight greater than $\Delta$ and performing a graph search from $(p_1, s_1)$ to $(p_n, s_m)$ takes $O(n m)$ time. The continuous decision problem admits a related construction that we discuss in the preliminaries.

\subparagraph{Computing the Fréchet distance and monotonicity events.}
Computing the \emph{value} of the Fréchet distance is considerably more intricate. The continuous distance is realised by one of three types of events (see Figure~\ref{fig:events}). Vertex events correspond to distances between pairs of vertices in $\pi \times \sigma$; edge events correspond to distances between a vertex of one curve and an edge of the other; and monotonicity events correspond to triples $(e, s, s')$, where $e$ is an edge of one curve, $s, s'$ are vertices of the other curve, and the corresponding distance value is realised by  the point on $e$ minimising the maximum distance to $s$ and $s'$. 
The discrete distance $\dd(\pi, \sigma)$ is always realised by a vertex event. As there are $O(nm)$ such events one may therefore compute all distances, sort them, and apply the decision procedure, yielding an $O(n m \log nm)$ algorithm (or $O(n m)$ with a simple dynamic program).
This version is also easy to make exact, as one can simply square the vertex-weights to preserve rationality of the input. 
For the continuous setting, monotonicity events introduce two significant difficulties:
\begin{itemize}
    \item There are $O(n^3 + m^3)$ such events, so no quadratic-time algorithm can enumerate them.
    \item Edge and monotonicity events involve equations that contain several square root monomials. They therefore offer no straightforward way to avoid numerical imprecision, even if the input is integral. This causes problems, not only for computing the distance value, but even for computing the correct combinatorial traversal that realises the Fréchet distance.
\end{itemize}
\noindent
An exact decision procedure for the continuous variant was provided by Bringmann, Künnemann, and Nusser~\cite[SoCG'19]{bringmann2019Walking}, who use exact geometric predicates to test $\fd(\pi, \sigma)\leq \Delta$ in $O(nm )$ time. Their implementation allows an approximation of $\fd(\pi,\sigma)$ to $b$ bits precision via binary search in $O(nm(\log nm + \log b))$ time, but it does not compute the exact value.

\begin{figure}[h]
    \centering
    \includegraphics[width=\linewidth]{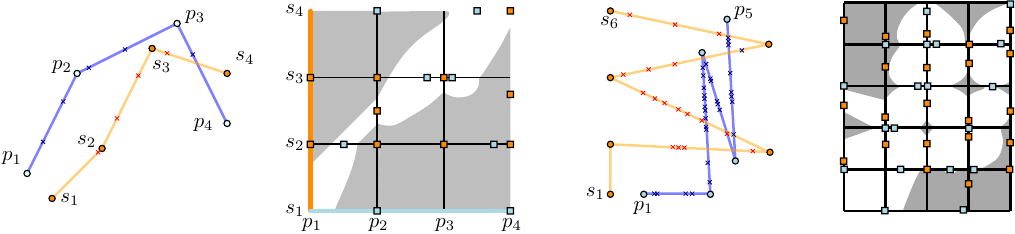}
    \caption{Two curves $(\pi, \sigma)$ with their parameter space. White space indicates that the corresponding points are within some distance $\Delta$.  In each grid, we illustrate the edge events. For each monotonicity event defined by an edge $e$ and vertices $(x_1, x_2)$ we show the event on $e$ with a cross. }
    \label{fig:events}
\end{figure}

\begin{figure}[t]
    \centering
    \includegraphics[width=\linewidth]{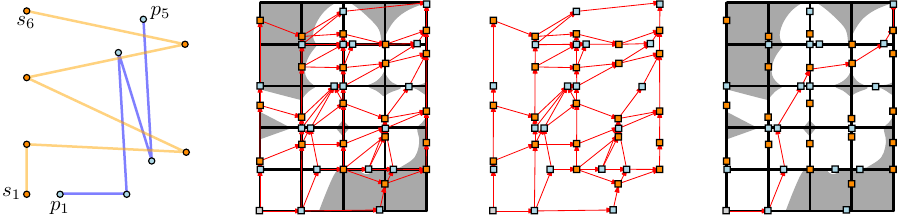}
    \caption{Two pairs of curves $(\pi, \sigma)$ and their \texttt{VE}-graph which is a directed and vertex-weighted graph. The minimum-cost path from $(1, 1)$ to $(n, m)$ in this graph is not a monotone traversal.  }
    \label{fig:ve_frechet}
\end{figure}
 
\subparagraph{Approaching the continuous Fréchet distance with a frog.}
Har-Peled, Raichel, and Robson~\cite[SoCG'25]{harpeled_et_al:LIPIcs.SoCG.2025.54} recently propose a different approach. Given $(\pi, \sigma)$ they construct what they call the \texttt{VE}-graph whose vertices correspond to all vertex and edge events.
In this graph, there is an edge between two vertices if they share a cell and either the edge is monotone, or the events lie on an edge of the same curve (see Figure~\ref{fig:ve_frechet}). 
The vertices are weighted by the event distance, and the shortest path from $(1, 1)$ to $(n, m)$ in this graph corresponds to a (non-monotone) person-dog traversal $\gamma$. If $\gamma$ happens to be $xy$-monotone, its continuous counterpart realises the Fréchet distance. Otherwise, non-monotonicities certify that the Fréchet distance is realised by a monotonicity event. Instead of computing monotonicity events, their algorithm bisects each edge of $\pi$ and $\sigma$ where $\gamma$ is non-monotone in the corresponding row or column.
This creates higher-complexity curves $(\pi', \sigma')$ on which they recurse. 
They call this \emph{refinement}, and its recursive application yields increasingly complex curves whose vertex and edge events approach the continuous Fréchet distance $\fd(\pi, \sigma)$. In the limit, this procedure converges to introducing the necessary monotonicity events and thereby  computes the continuous Fréchet distance through discrete shortest paths.

Their second contribution is a clever simplification technique that is agnostic to how the Fréchet distance between curves is computed. 
They show how to simplify two input curves $(\pi, \sigma)$ into lower-complexity curves $(\overline{\pi}, \overline{\sigma})$ such that $\fd(\pi, \sigma) = \fd(\overline{\pi}, \overline{\sigma})$. 
The procedure is iterative: they first compute simplified curves $(\overline{\pi}, \overline{\sigma})$ and evaluate $\fd(\overline{\pi}, \overline{\sigma})$ using any method. This yields a person-dog traversal, from which they  generate a witness certifying whether $\fd(\overline{\pi}, \overline{\sigma}) = \fd(\pi, \sigma)$ holds. 
If not, the curves are partially un-simplified, reintroducing complexity until the criterion is satisfied.
Based on our observations, this simplification step is likely the main source of the speedy performance of their implementation.

\subparagraph{Applications of Fréchet distance.}
The Fréchet distance has many well-known applications, in particular in the analysis and visualisation of movement data~\cite{buchin2017clustering, buchin2020group, konzack2017visual, xie2017distributed}. 
It is a versatile measure used across a wide range of domains, including handwriting~\cite{sriraghavendra2007frechet}, coastlines~\cite{mascret2006coastline}, geographic outlines~\cite{devogele2002new}, trajectories of vehicles, animals, and sports players~\cite{acmsurvey20, su2020survey, brakatsoulas2005map, buchin2020group}, air traffic~\cite{bombelli2017strategic}, and protein structures~\cite{jiang2008protein}. 
These typically provide their own Fréchet distance implementations.
A recent and particularly active domain is \emph{trajectory clustering}, where the task is to group (sub)curves into clusters such that each cluster has low pairwise Fréchet distance. 
Several approaches adopt this principle, including the subtrajectory clustering algorithm of Agarwal et al.~\cite{agarwal2018subtrajectory}, the map construction and clustering algorithms of Buchin et al.~\cite{buchin2017clustering, buchin2020improved}, the clustering method of Buchin et al.~\cite{buchin2011detecting} (implemented in the \texttt{MOVETK} library~\cite{MOVETK}), the clustering algorithm of Conradi and Driemel~\cite{conradi2023finding}, and the clustering of van der Hoog et al.~\cite{vanderhoo2025Efficient}. 
In the \emph{vast} majority of these examples (clustering and otherwise), the implementation computes the discrete Fréchet distance, as this is considerably easier to evaluate in practice.

Two of these approaches~\cite{buchin2017clustering, buchin2020improved} compute the \emph{semi-weak} Fréchet distance, in which the person and the dog are not required to be monotone within an edge. This variant corresponds to the edge-monotone minimum-cost traversal computed by Har-Peled, Raichel, and Robson~\cite{harpeled_et_al:LIPIcs.SoCG.2025.54} before their curve refinement procedure is applied. 
These examples illustrate the need for efficient and adaptable Fréchet distance implementations, and highlight why the refinement-based approach of Har-Peled, Raichel, and Robson is appealing: it suggests the possibility of computing the exact continuous Fréchet distance through a discrete traversal mechanism that is compatible with existing algorithmic paradigms and interfaces.

\subparagraph{Contribution.}
We re-examine the frog-based framework of Har-Peled, Raichel, and Robson~\cite{harpeled_et_al:LIPIcs.SoCG.2025.54}, as we revisit three aspects of their contribution. 

\subparagraph{Contribution 1: exact-value computations.} We find the over-arching method conceptually very interesting (and even elegant), however, their algorithm does not compute the exact Fréchet distance. 
They state that ``as the calculation of the optimal Fréchet distance involves distances, impreciseness is unavoidable''. In practice, there exist exact-evaluation predicates. In theory, one can even assume a real RAM. Their current approach would remain imprecise under such assumptions:
\begin{enumerate}
    \item Their implementation inherently uses floating-point arithmetic and thus impreciseness.
    \item However, even with exact real-valued arithmetic, their edge bisection does not introduce the monotonicity events that realise the exact Fréchet distance. This refinement converges to the correct value only in the limit, with no bound on the number of required iterations.
   \item Their lossless simplification step is theoretically sound and would yield the exact Fr\'echet distance under real-valued arithmetic. To ensure that the Fr\'echet distance remains unchanged after simplification, they impose constraints on how the curves may be simplified which can be relaxed. 
\end{enumerate}
The original paper briefly acknowledges the possibility of an exact algorithm.
Our motivation is to understand to what extent this framework \emph{can} be made exact, and to provide the required components to do so.
\begin{enumerate}
    \item We explain in detail how to unpack square-root expressions for exact evaluation. This is non-trivial, particularly for monotonicity events and curve simplifications where several square-root terms must be handled simultaneously.
    \item We replace their heuristic refinement procedure with a new output-sensitive algorithm that generates the required monotonicity events on demand. 
    As a consequence, we obtain provable polynomial-time convergence, even under real-valued computations.
    \item We strengthen the lossless simplification technique. Whilst the method from~\cite{harpeled_et_al:LIPIcs.SoCG.2025.54} works as a black-box given any Fréchet computation algorithm that outputs a traversal, we note that this frog-based paradigm allows for a tighter algorithm and analysis. This leads to an alternative and new algorithm for lossless simplification under the frog-based method.  
\end{enumerate}

\subparagraph{Contribution 2: exact implementation.}
We provide an open-source C++ implementation of our adapted algorithm with our exact arithmetic and monotonicity event-generation. This  is thereby the first exact implementation. We experientially compare our solution to the one provided in~\cite{harpeled_et_al:LIPIcs.SoCG.2025.54}. Unsurprisingly, our median running times increase due to exact computation. Our exact-value representations simply do now allow for vectorization, which comes at performance cost. Somewhat unexpectedly, the worst-case running time drastically improves.
We can show instances where the recursive refinement encounters its worst-case behaviour and terminates with an approximation only after many refinement steps. 
For some data sets, the worst-case behaviour is such that even our average time is an improvement. 

\subparagraph{Contribution 3: extensive empirical analysis.}
Our final contribution is the most significant. 
\cite{harpeled_et_al:LIPIcs.SoCG.2025.54} claims practical performance gains over prior Fréchet distance computation implementations. 
We observe that the empirical evaluation in~\cite{harpeled_et_al:LIPIcs.SoCG.2025.54} is conducted on animal-tracking data, which exhibit two properties: animal trajectories are goal-oriented (which leads to little to no backtracking), and they frequently contain long runs of near-duplicate vertices (as animals stand still for prolonged periods of time). These characteristics favour the lossless simplification step in~\cite{harpeled_et_al:LIPIcs.SoCG.2025.54}, as they can effectively reduce the input size before the main computation. We perform a broader evaluation by considering additional data sets: un-edited real-world data set from traffic data and adversarial inputs derived from orthogonal vectors.
Across all our data sets, we observe a substantially different performance profile than in~\cite{harpeled_et_al:LIPIcs.SoCG.2025.54}. In particular, the highly-optimised implementation of Bringmann, Künnemann, and Nusser \texttt{BKN}~\cite{bringmann2019Walking} dominates all frog-based implementations in performance (e.g., Figure~\ref{fig:results}).

At first sight, our empirical analysis may appear unfavourable to our contribution, as it shows that on many practical instances the performance of this new frog-based approach is surpassed by classical techniques.
However, we argue for our significance:  We are the first exact Fréchet distance implementation. Moreover, the empirical analysis reported at SoCG'25~\cite{harpeled_et_al:LIPIcs.SoCG.2025.54} claimed strong practical performance over~\cite{bringmann2019Walking}.
Our broader evaluation provides a complementary picture: it clarifies that this technique is valuable, but also clearly shows the limits of its use. 
Ultimately, we argue that the frog-based technique of~\cite{harpeled_et_al:LIPIcs.SoCG.2025.54} is of genuine theoretical interest, as it can be augmented to compute the exact Fréchet distance in a unique manner with provable guarantees. 
At the same time, its practical feasibility is debatable.

\begin{figure}[h]
    \centering
    \includegraphics[width=\textwidth]{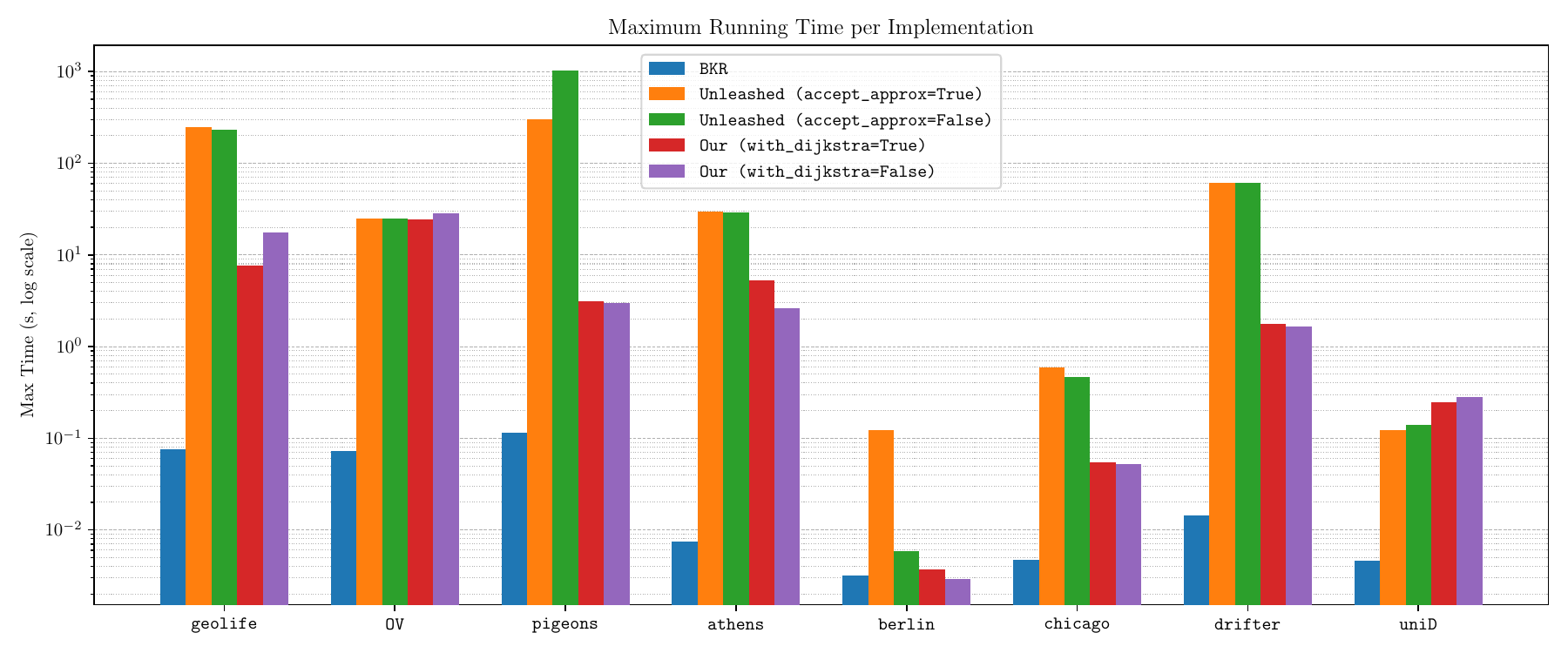}
    \caption{Maximum running times in seconds, log-scale, for curves from eight data sets.}
    \label{fig:results}
\end{figure}

\newpage
\section{Preliminaries}

The input consists of polygonal curves $\pi$ and $\sigma$.
A \emph{curve} $\pi = (p_1, \ldots, p_n)$ is an ordered sequence of \emph{vertices} connected by straight-line \emph{edges}.  
We implicitly endow every $n$-vertex curve $\pi$ with a bijective, monotone parametrisation $\pi : [1,n] \to \mathbb{R}^d$ such that $\pi(i) = p_i$ for all integers $i$.  
For $t \in [i,i+1]$, the mapping is defined by linear interpolation $
  \pi(t) := p_i + (p_{i+1}-p_i)\, (t-i)$.
  
  For an ordered pair of curves $(\pi,\sigma)$ of sizes $(n,m)$, a \emph{traversal} is any pair of continuous, monotone mappings
  $\alpha : [0,1] \to [1,n]$ and $\beta : [0,1] \to [1,m]$.  
  The (continuous) Fréchet distance $\fd(\pi,\sigma)$ is defined as
  $
    \fd(\pi,\sigma)
      := \min\limits_{\text{traversals } (\alpha,\beta)}\;
           \max\limits_{t \in [0,1]} d\bigl( \pi(\alpha(t)),\, \sigma(\beta(t)) \bigr)\, .
  $

Equivalently, consider the parameter space $\Gamma_{\pi,\sigma} := [1,n] \times [1,m]$, subdivided into an $n \times m$ grid.  
A traversal corresponds to any continuous, $xy$-monotone curve $\gamma(t) = (\alpha(t),\beta(t))$ from $(1,1)$ to $(n,m)$ in $\Gamma_{\pi,\sigma}$.  
We therefore freely consider traversals as monotone curves.

\subparagraph{Deciding.}
We can decide whether $\fd(\pi,\sigma) \le \Delta$ as follows.  
For each point $(a,b) \in [1,n] \times [1,m]$, define $\phi(a,b)=0$ if and only if $d\bigl(\pi(a),\sigma(b)\bigr) \le \Delta$.  
The area $\{(a,b) \mid \phi(a,b)=0\}$ is the \emph{free space}.   
A straightforward DFS over the free space decides in $O(nm)$ time whether the free space contains a monotone curve from $(1,1)$ to $(n,m)$, and hence whether $\fd(\pi,\sigma) \le \Delta$.

\subparagraph{Computing.}
The Fréchet distance is realised by one of three types of \emph{events}:
\begin{itemize}
  \item \emph{Vertex events}: distances between pairs of vertices in $\pi \times \sigma$.
  \item \emph{Edge events}: distances between a vertex of one curve and an edge of the other.
  \item \emph{Monotonicity events}: triples $(e,s,s')$ where $e$ is an edge of one curve and $s,s'$ are vertices of the other, corresponding to minimising the maximum distance to $s$ and $s'$.
\end{itemize}
A naive algorithm enumerates all events in $O(n^3 + m^3)$ time, sorts them, and performs a binary search with the decision algorithm to obtain $\fd(\pi, \sigma)$.  
Alt and Godau~\cite{alt1995computing} observed that many monotonicity events are irrelevant and achieved an $O(nm \log nm)$-time exact algorithm.  
Har-Peled and Raichel~\cite{HarPeled2011Frechet} gave a considerably simpler approach:  
as $\Delta$ increases, connected components of the free space merge, and a \emph{critical value} is the smallest $\Delta$ at which a new merge occurs.  
Each cell has one connected component, implying $O(nm)$ critical values.  
By comparing adjacent cells, all critical values can be generated in $O(nm \log nm)$ time.  
The fastest algorithm is by Buchin et al.~\cite{Buchin2014Four}, running in  
$   O\!\left(n^{2}\, \sqrt{\log n}\, (\log\!\log n)^{3/2} \right) $ time.

\subparagraph{Retractability.}
For curves $(\pi,\sigma)$ and a traversal $\gamma = \{\gamma(t) \mid t \in [0,1]\}$, define $\mathrm{cost}(\gamma) := \max_{(a,b) \in \gamma} d\bigl(\pi(a),\sigma(b)\bigr).$
Let $(a,b) \in \gamma$ be a point that attains this maximum.  
Then $\gamma$ decomposes into two open curves:  
$\gamma_1$ from $(1,1)$ to $(a,b)$ and $\gamma_2$ from $(a,b)$ to $(n,m)$.  
Har-Peled, Raichel and Robson~\cite{harpeled_et_al:LIPIcs.SoCG.2025.54}  call $\gamma$ \emph{retractable} if  
$\gamma_1$ and $\gamma_2$ are minimum-cost monotone curves for their respective endpoints,  
and if $\gamma_1$ and $\gamma_2$ are themselves retractable (Figure~\ref{fig:retractable}).
Computing the Fréchet distance using an adaption of Dijkstra yields a retractable traversal. 

\begin{figure}[h]
  \centering
  \includegraphics[width = \linewidth]{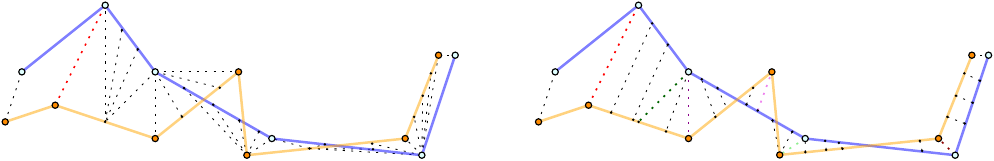}
  \caption{(Left) A traversal between two curves where the red leash realises $\fd(\pi,\sigma)$.  
  The traversal may be `lazy': a large Fréchet distance allows long pauses while the other curve progresses.  
  (Right) A retractable traversal achieves progressively smaller bottleneck leash lengths.}
  \label{fig:retractable}
\end{figure}
\newpage

\section{The frog-based approach}
We paraphrase and explain the frog-based approach from~\cite{harpeled_et_al:LIPIcs.SoCG.2025.54}, defining their \texttt{VE}-graph slightly differently to reduce the number of edge-cases the definition incurs. We then detail how to make it converge in a polynomial number of rounds, even on a real-RAM, by creating an algorithm that generates the required monotonicity events on the fly.

\begin{definition}
   Each edge $e$ of the grid  $\Gamma_{\pi, \sigma}$ has an edge event defining a point that we call an \emph{edge-distance yardstick} (or \emph{eddy}).
$V(\pi, \sigma)$ denotes the set of all grid vertices and eddys. 
\end{definition}

\begin{figure}[b]
    \centering
    \includegraphics[width = \linewidth]{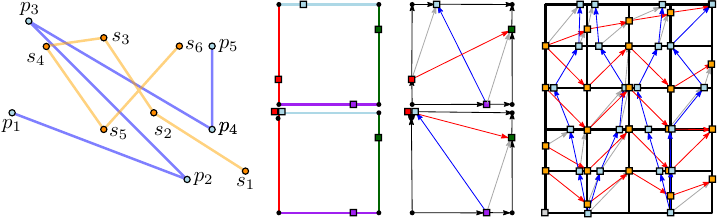}
    \caption{ 
    The \texttt{VE} graph considers for each cell the four corners and the four eddys. 
    We include all monotone edges between these vertices. 
    Edges connecting the eddys between horizontal (or vertical) edges may be non-monotone. If vertices coincide, then the picture gets a bit messy. }
    \label{fig:ve_graph}
\end{figure}

\begin{definition}
  The \texttt{VE}-graph is a directed vertex-weighted graph on $V(\pi,\sigma)$.  
  For each cell of $\Gamma_{\pi,\sigma}$, consider its four corners and eddys. We connect all these points in an $xy$-monotone manner.  
  In addition, the bottom horizontal eddy has a directed edge to the top horizontal eddy, and the left vertical eddy has a directed edge to the right vertical eddy; see Figure~\ref{fig:ve_graph}.
\end{definition}

\begin{definition}
  A path in the \texttt{VE}-graph is \texttt{VE}-respecting if it runs from $(1,1)$ to $(n,m)$ and, within every row and every column, its restriction to the row or column is connected.  
  \[
    \ve(\pi,\sigma)
      := \min_{\texttt{VE}\text{-respecting path } P}
           \; \max_{\text{vertex } v \in P}
           \mathrm{weight}(v).
    \qquad\qquad
    \text{Observe that } \ve(\pi,\sigma) \le \fd(\pi,\sigma).
  \]
\end{definition}

\noindent
From a path realizing the \texttt{VE} distance, one may also derive an upper bound on $\fd(\pi, \sigma)$.

\begin{definition}
    Let $P$ be a path realizing $\ve(\pi,\sigma)$ then $P$ is a curve. One can greedily morph $P$ into a traversal by ensuring that one never decreases in either coordinate. This defines an $xy$-monotone traversal $\gamma(t) = (\alpha(t), \beta(t))$.
    The \emph{Interpolated \texttt{VE} distance} is then defined as:
    \[
        \ive(\pi, \sigma)
        := \max_{t \in [0,1]} d\bigl(\pi(\alpha(t)), \sigma(\beta(t))\bigr).
        \qquad\qquad
        \text{Observe that } \fd(\pi,\sigma) \le \ive(\pi,\sigma).
    \]
\end{definition}

The \emph{elevation function} maps $(a,b) \in \Gamma_{\pi, \sigma}$ to $d(\pi(a), 
\sigma(b))$.
The sublevel sets of this function are clipped ellipses, whose minima on each cell edge occur precisely at the corresponding eddys.
It follows that $\ve(\pi,\sigma) = \fd(\pi, \sigma) = \ive(\pi, \sigma) $ whenever the path $P$ realizing $\ve(\pi,\sigma)$ is monotone.
This is the core of the frog-based refinement approach (Algorithm~\ref{alg:refine}).

\begin{algorithm}[H]
\caption{Refinement Procedure from~\cite{harpeled_et_al:LIPIcs.SoCG.2025.54}}
\label{alg:refine}
\begin{algorithmic}[1]
  \State Compute the \texttt{VE}-graph, the value $\ve(\pi,\sigma)$, and the path $P$ realising it
  \State \textbf{if} $P$ is monotone \textbf{then} $\ve(\pi,\sigma) = \fd(\pi, \sigma)$ so \Return $\ve(\pi,\sigma)$
  \State Otherwise, insert vertices on $\pi$ for every non-x-monotone column, and vertices on $\sigma$ for every non-$y$-monotone row. Then recurse on the refined curves $(\pi,\sigma)$.
\end{algorithmic}
\end{algorithm}

\subsection{The approach in~\cite{harpeled_et_al:LIPIcs.SoCG.2025.54}: bisecting and convergence.}
Har-Peled, Raichel, and Robson~\cite{harpeled_et_al:LIPIcs.SoCG.2025.54} execute step 3 of Algorithm~\ref{alg:refine} in a heuristic manner.  
Let $P$ be a path realising $\ve(\pi, \sigma)$ for the current input $(\pi, \sigma)$. 
If $P$ is non-monotone in a column that corresponds to an edge $e$ of $\pi$ then they introduce a vertex half-way $e$ (likewise, they bisect edges of $\sigma$ for each row where $P$ is non-monotone). 
Thus, given $P$ they bisect some edges and create higher-complexity curves $(\pi, \sigma)$ on which they recurse.

Recall that each monotonicity event $(e,p,p')$ corresponds to a point on the edge $e$; we refer to such a point as a \emph{monotonicity vertex}.  
Let $(\pi',\sigma')$ be the curves obtained from $(\pi,\sigma)$ by inserting, on $\pi$ and $\sigma$ respectively, a vertex at every monotonicity vertex.
It is observed in~\cite{harpeled_et_al:LIPIcs.SoCG.2025.54} that the path $P$ realising $\ve(\pi',\sigma')$ is $xy$-monotone, and therefore $
  \ve(\pi',\sigma') = \fd(\pi,\sigma).$
Under exact-value arithmetic, the recursive halving procedure of~\cite{harpeled_et_al:LIPIcs.SoCG.2025.54} inserts in the limit all monotonicity vertices of $\pi$ and $\sigma$, and thus converges to the exact Fréchet distance $\fd(\pi,\sigma)$.

\subsection{Our approach and rate of convergence.}

In~\cite{harpeled_et_al:LIPIcs.SoCG.2025.54} they use floating-point arithmetic. 
Therefore,  edges cannot be bisected in an exact manner and their bisecting never generates the exact monotonicity vertices. 
They remark that under exact computations one can compute and introduce monotonicity vertices instead, but do not provide details, and naively there are cubicly many such events. 
Floating-point arithmetic is not unavoidable. Many libraries support exact arithmetic (e.g., CGAL). We engineer a solution that alters step 3 to provide guarantees for exactness and convergence. 

Let the path $P$ realising $\ve(\pi, \sigma)$ be not $x$-monotone in a column $C$ corresponding to an edge $e$. 
Let the subpath of $P$ restricted to $C$ be from $a$ to $b$. We execute step 3 by computing the minimum number of monotonicity vertices that need to be introduced on $e$ such that, next iteration, the min-cost path from  $a$ to $b$ is monotone.
Our algorithm is inspired by the critical values from~\cite{alt1995computing, HarPeled2011Frechet} from the preliminaries, but ultimately has a different objective.

\subparagraph{On efficiency.}
Exact-value arithmetic inevitably incurs overhead.  
Evaluating exact geometric predicates is substantially slower than using floating-point arithmetic. Moreover, the resulting values must be maintained as abstract algebraic objects rather than bit-level representations, preventing efficient vectorisation. This is partly why the algorithm from~\cite{bringmann2019Walking} does not use exact arithmetic even though it is integrated in CGAL. 
We therefore expect exact computation to increase median running times when we introduce exact arithmetic. 
To mitigate this effect, we restrict the input to $32$-bit integers, as this allows the most frequent operation---the computation of dot products---to still be efficient. For our exact arithmetic, based on the computed dot products, we discuss in Appendix~\ref{sec:exact} our own integer-based kernel, to remain practically somewhat competitive. 
Importantly, our algorithm does not rely on integrality assumptions: it applies unchanged under any exact-evaluation kernel.

\subparagraph{High-level overview.}
In appendix~\ref{app:monotonicity} we develop the following algorithm:
the input is a column $C$ where $P$ is not $x$-monotone, corresponding to an edge $e$ of $\pi$. Let $P$, restricted to $C$, go from $a$ to $b$. We create two new cells by extending a horizontal segment from $a$ and $b$. We discard all cells of $C$ that do not intersect $P$. We index the remaining cells as $C_1, \cdot, C_k$ from bottom to top, which induce horizontal segments $h_1, \ldots, h_{k+1}$. Imagine $\Delta$ increasing from $0$ to $\infty$.
For each value of $\Delta$ and each pair $(i, j)$ with $i < j$, $h_i$ can \emph{$\Delta$-reach} $h_j$ if there exists a monotone path from a point on $h_i$ to a point on $h_j$ in the $\Delta$ free space (see Figure~\ref{fig:shootup}).

A value $\Delta$ is an \emph{join event} if there exist $(i,j)$ such that $h_i$ can $\Delta$-reach $h_j$ from $\Delta$ onward.  
In this case, $h_i$ can also $\Delta$-reach all $h_{j'}$ with $j' \in [i,j)$.
We compute the minimum number of join events to realise a monotone path from $h_1$ to $h_{k}$ using a linked–list of buckets.

The horizontal segments are partitioned into buckets, each forming a vertical interval $[i,j]$ such that $h_i$ can $\Delta$-reach $h_j$.
We initialise $\Delta$ as $0$, and every horizontal segment forms its own bucket.
We then merge buckets in a bottom-up fashion, maintaining the minimum $\Delta$ required such that $h_1$ can $\Delta$-reach $h_i$ for increasing $i$. 
Between $(h_1, \ldots, h_i)$ and $h_{i+1}$, we compute the smallest value $\Delta' \geq \Delta'$ for $h_1$ can $\Delta$-reach $h_{i+1}$ using a min-heap that contains join events and other events that indicate for which pair $(j, i+1)$ the join event is relevant.   

Whenever we pop the next event, we update $\Delta$.
If the updated $\Delta$ allows us to $\Delta$-reach $h_{i+1}$ from $h_1$ then we merge $h_{i+1}$ into the bucket.
We then keep incrementing $i$ until $h_1$ cannot $\Delta$-reach $h_{i+1}$, at which point we continue the event-based approach. 
Whenever we merge $h_k$ into the bottom bucket, $h_1$ can $\Delta$-reach $h_k$. 
When all buckets are joined, we have all join events required for a monotone path from $h_1$ to reach $h_k$. From this,  we obtain the minimum number of monotonicity vertices needed so that, in the next iteration, the min-cost path from $a$ to $b$ is monotone. Doing this for all non-monotone rows and columns implies:

\begin{restatable}{theorem}{shootup}
Let $(\pi,\sigma)$ be the original input curves with $n$ and $m$ vertices. Algorithm~\ref{alg:refine} can be supported by a data structure using $O(n+m+M)$ space, where $M \in O(n^{3} + m^{3})$ is the number of introduced monotonicity vertices. Each recursive call runs in $O(nm + M)$ time for Step~1 and $O( (n+m+M) \log^2 (n+m+M))$ time for Step~3. Under exact-value arithmetic, the procedure returns the exact Fréchet distance after $O(n^{3} + m^{3})$ recursive calls.
\end{restatable}

\begin{figure}
    \centering
    \includegraphics[width  = \linewidth]{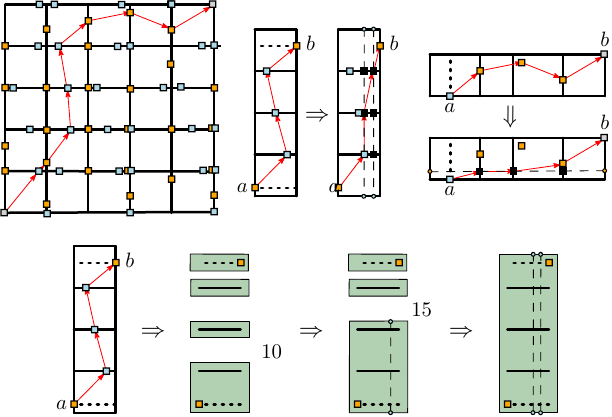}
    \caption{The input is a non-monotone path $P$ in the \texttt{VE}-graph.
We identify all rows and columns where $P$ is non-monotone.
For a column, we isolate the horizontal segments and bucket them. Between buckets, we store the minimum leash at which a monotone path from bottom to top exists.
    }
    \label{fig:shootup}
\end{figure}

\newpage
\section{Lossless Fréchet distance simplification}

The second (perhaps even main) contribution of~\cite{harpeled_et_al:LIPIcs.SoCG.2025.54}  is a black-box lossless simplification algorithm. 
Computing $\fd(\pi, \sigma)$ for large complexity curves $\pi$ and $\sigma$ is an expensive operation.
It would be preferable if we would have access to some lower-complexity curves $(\overline{\pi}, \overline{\sigma})$ where $\fd(\pi, \sigma) = \fd(\overline{\pi}, \overline{\sigma})$. Then we could simply operate on these lower-complexity curves to output the distance. 
Har-Peled, Raichel, and Robson~\cite{harpeled_et_al:LIPIcs.SoCG.2025.54}  can, for any Fréchet computation algorithm, achieve this using any \emph{vertex-restricted} curve simplification: 

\begin{definition}
    We define for $\pi$ a (vertex-restricted) simplification $\overline{\pi}$ 
    as any subset of vertices in $\pi$, ordered by their index.    
    For any $\mu \in \mathbb{R}_{\geq 0}$ we say $\overline{\pi}$ is a $\mu$-simplification if $\fd(\pi, \overline{\pi}) \leq \mu$.
\end{definition}

\noindent
The curve simplification used in~\cite{harpeled_et_al:LIPIcs.SoCG.2025.54} is some heuristic combination of techniques from~\cite{AronovHKWW06, DriemelHW10}.
\noindent

Naïvely, one can $(1+\varepsilon)$-approximate $\fd(\pi, \sigma)$ with simplifications by fixing some $\varepsilon$ and $\mu$. Compute $\mu$-simplifications $(\overline{\pi},\overline{\sigma})$ and set $M=\fd(\overline{\pi},\overline{\sigma})$. 
By the triangle inequality, $
M-2\mu \;\leq\; \fd(\pi,\sigma) \;\leq\; M+2\mu.$
If $\mu\leq \varepsilon M$, then $M$ $(1+\varepsilon)$-approximates $\fd(\pi,\sigma)$; otherwise $\mu$ is decreased cleverly until we terminate   after $O(\log_{1+\varepsilon} n)$ rounds~\cite{DriemelHW10}.

\cite{harpeled_et_al:LIPIcs.SoCG.2025.54} avoids this approximation step by exploiting triangle inequalities more directly. 
Let $(\overline{\pi},\overline{\sigma})$ be any vertex-restricted simplification of $(\pi,\sigma)$ and consider the traversals $\gamma_0(t)$, $\gamma_1(t)$, and $\gamma_2(t)$ realising $\fd(\pi,\overline{\pi})$, $\fd(\overline{\pi},\overline{\sigma})$, and $\fd(\overline{\sigma},\sigma)$, respectively. 
They combine these traversals in linear time into a traversal $\gamma(t)=(\alpha(t),\beta(t))$ between $(\pi,\sigma)$ whose maximum leash length gives an upper bound on $\fd(\pi,\sigma)$. 
Conversely, $\fd(\pi,\sigma)$ is at least $\fd(\overline{\pi},\overline{\sigma})$ minus the Fréchet distances from each curve to its simplification, and we get the following:
\[
\fd(\overline{\pi},\overline{\sigma}) - \fd(\pi,\overline{\pi}) - \fd(\sigma,\overline{\sigma})
\;\leq\; \fd(\pi,\sigma)
\;\leq\;
\max_{t\in[0,1]} \fd(\pi(\alpha(t)),\sigma(\beta(t))).
\]

Computing $\fd(\pi,\overline{\pi})$ and $\fd(\sigma,\overline{\sigma})$ is costly, so they compute greedy traversals $\gamma_0^*(t)$ and $\gamma_2^*(t)$ that over-estimate these values instead, yielding $\fd^*(\pi,\overline{\pi})$ and $\fd^*(\sigma,\overline{\sigma})$. 
Combined, these give a greedy traversal $\gamma^*(t)=(\alpha^*(t),\beta^*(t))$ that over-estimates $\fd(\pi,\sigma)$ and leads to
\begin{equation}
\label{eq:inequality}
LB(\bar\pi,\bar\sigma)
:=\fd(\bar\pi,\bar\sigma)-\fd^*(\pi,\bar\pi)-\fd^*(\sigma,\bar\sigma)
\le \fd(\pi,\sigma)
\le \max_{t\in[0,1]}\fd(\pi(\alpha^*(t)),\sigma(\beta^*(t))).
\end{equation}

\subparagraph{Introducing slack.}
Given $LB(\overline{\pi},\overline{\sigma})$ and $\gamma^*(t)$, each curve vertex $v$ receives a \emph{slack} $
\texttt{slack}(v) := LB(\overline{\pi},\overline{\sigma}) - g_v$,
where $g_v$ is the maximum leash length attained at $v$ along $\gamma^*(t)$ (Figure~\ref{fig:slack}).  
Negative slack marks $v$ as a bottleneck.  
Each vertex $p$ of $\overline{\pi}$ or $\overline{\sigma}$ inherits the minimum slack of all original vertices between $p$ and the next vertex.  
Slack is then heuristically propagated: every vertex adopts the minimum slack among its eight closest neighbours along $\overline{\pi}$ or $\overline{\sigma}$. 
An edge’s slack is the minimum slack of its endpoints.  
They then perform the following:

\begin{quote}
    while there exists a negative-slack edge $e$ on $\overline{\pi}$ (resp.\ $\overline{\sigma}$), and the edge $e$ is not an edge of $\pi$ (resp.\ $\sigma$), un-simplify the edge $e$ (reintroducing missing vertices).
\end{quote}

\noindent
By construction of slack, its propagation, and~\eqref{eq:inequality}, the process ends with $\fd(\pi,\sigma)=\fd(\overline{\pi},\overline{\sigma})$.

\subparagraph{Our contribution.}
Appendix~\ref{sec:lossless} shows that the above technique, while elegant, is overly conservative and somewhat coarse.  
We refine it in two ways.

\clearpage

\begin{figure}[H] \centering \includegraphics[]{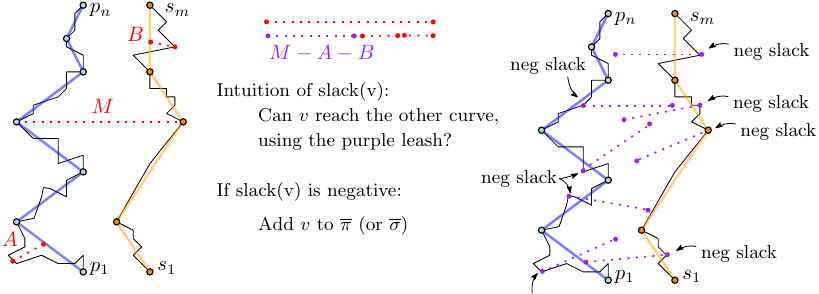} \caption{Curves $(\pi, \sigma)$ and their simplifications. The value $LB(\overline{\pi}, \overline{\sigma})$ is the maximum leash length $M$ between $(\overline{\pi}, \overline{\sigma})$, minus the maximum leash lengths $A$ and $B$ between $(\pi, \overline{\pi})$ and $(\sigma, \overline{\sigma})$.} \label{fig:slack} \end{figure}

\medskip\noindent\emph{(i) Tighter bounds.}
Their lower bound uses the global maximum leash length $M$ between $(\overline{\pi},\overline{\sigma})$, minus the global maxima $A$ and $B$ between $(\pi,\overline{\pi})$ and $(\sigma,\overline{\sigma})$, respectively.  
When $M$, $A$ and $B$ are realised at unrelated parts of the traversal (see Figure~\ref{fig:slack}), this yields an unnecessarily small value $LB(\overline{\pi}, \overline{\sigma}) = M-A-B$. This in turn causes more vertices than needed to trigger unsimplification. 
Using the frog-based framework, we obtain a sharper bound (Figure~\ref{fig:weighted}).  
For each edge of $\overline{\pi}$ we compute the maximum leash length between that edge and $\pi$, and assign this as a negative weight.  
In the \texttt{VE}-graph, the eddys on this edge add this negative weight to their own weight.  
The min-cost path in this reweighted graph yields a lower bound for $\fd(\pi,\sigma)$ by the same triangle-inequality argument as in~\cite{harpeled_et_al:LIPIcs.SoCG.2025.54}, but now based on local rather than global maxima.  
For many instances, this increases the lower bound and therefore the slack of vertices --- causing our algorithm to unsimplify fewer edges. 

\medskip\noindent\emph{(ii) A tighter slack propagation rule.}
Their propagation step takes the minimum slack over each block of eight consecutive vertices. This removes any edge cases for applying the slack, ensuring their correctness. However, this method is somewhat crude.
We tighten this step by propagating slack along at most two edges per vertex, avoiding this coarse eight-neighbour spread. 
\begin{figure}[H] \centering \includegraphics[width = 0.85\linewidth]{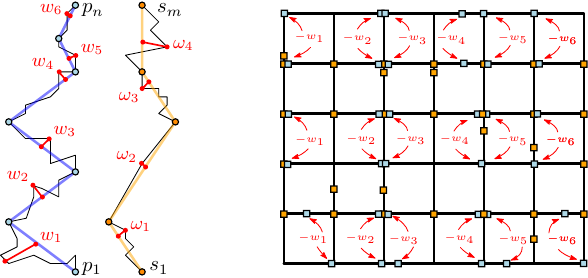} \caption{Given $(\pi, \overline{\pi})$ we can give a weight to each edge of $\overline{\pi}$. Each edge over $\overline{\pi}$ is a column in the parameter space, in which we apply its negative weight on the corresponding eddys.} \label{fig:weighted} \end{figure}

\newpage

\section{Experiments and empirical analysis.}

Our primary contribution is a C++ implementation of our exact variant of the frog-based technique, together with an extensive empirical evaluation. 
We provide two implementations: 
\texttt{Exact\_Dijkstra}, which computes the path $P$ in the \texttt{VE}-graph using Dijkstra’s algorithm in a manner similar to~\cite{harpeled_et_al:LIPIcs.SoCG.2025.54} and yields a traversal with retractability properties; and 
\texttt{Exact\_Sweepline}, which instead employs a sweepline to compute $P$, guaranteeing linear space usage. 
We compare our implementations to \texttt{Unleashed}: the Julia code of~\cite{harpeled_et_al:LIPIcs.SoCG.2025.54}, which is their fastest version. 
To obtain the most accurate output from \texttt{Unleashed}, we invoke \texttt{frechet\_c\_compute} with \texttt{accept\_approximation=false} (this still approximates). 
Finally, we include the \texttt{BKN} algorithm of~\cite{bringmann2019Walking}, the fastest known alternative implementation (which is also inexact). We do not use its CGAL variant, as it is slower, while still being inexact.

\subparagraph{Experiments.}
The analysis in~\cite{harpeled_et_al:LIPIcs.SoCG.2025.54} uses three animal-tracking data sets: \texttt{Geolife}, \texttt{Birds}, and \texttt{Pigeons}. 
These trajectories are goal-oriented. 
Their implementation benefits from this in one major way compared to~\cite{bringmann2019Walking}: as a simplification is applied, this greatly reduces complexity. 
Consequently, they operate on significantly smaller effective input.

The lossless simplification of~\cite{harpeled_et_al:LIPIcs.SoCG.2025.54} applies to any Fréchet algorithm, including~\cite{bringmann2019Walking}. 
We \emph{could} enable a direct comparison by applying it to~\cite{bringmann2019Walking}. Yet, we deliberately refrain from doing so: the simplification incurs overhead, and~\cite{bringmann2019Walking} is so fast that this may actually slow it down. 
Despite not using simplifications, our experiments will show that~\cite{bringmann2019Walking} dominates the frog-based approaches, removing the need to improve it further.
All experiments are single-threaded and executed on an AMD Ryzen 5 7600X (4.7 GHz) with 96\,GB RAM.

\subparagraph{Data.}
Our real-world data sets are summarised in Tables~\ref{tab:data_sets} and \ref{tab:data_sets_slice}. 
The first two (\texttt{Geolife}, \texttt{Pigeons}) are from~\cite{harpeled_et_al:LIPIcs.SoCG.2025.54}. 
We additionally include \texttt{Athens}, \texttt{Chicago}, and \texttt{Berlin} GPS traffic data sets, widely used in Fréchet benchmarks~\cite{ahmed2015comparison, buchin2017clustering, huang2018automatic, wang2015efficient}. 
We include the \texttt{Drifter} data set of~\cite{conradi2023finding}, containing drifting-buoy trajectories.
Finally, we add the UnID data set from~\cite{9827305} which is a traffic dataset captured on German highways.
Except for \texttt{Drifter}, these sets are goal-oriented with minimal backtracking; lossless simplification therefore removes relatively many vertices, which favours the frog-based approaches.
Finally, we include a set $\texttt{OV}$\footnote{obtained via personal communication with one author from \cite{bringmann2019Walking}, available together with our source code} of difficult synthetic pairs of curves that are derived from Orthogonal Vectors instances~\cite{bringmann2019Walking}. 

The \texttt{Geolife} data set alone induces over $170{,}000{,}000$ pairs of curves, so we cannot possibly compute the Fréchet distance between all pairs of curves. 
We follow~\cite{harpeled_et_al:LIPIcs.SoCG.2025.54} and subsample. 
In our case, we compare all $45{,}000$ pairs among a subselection of $300$ curves from the \texttt{Geolife} and \texttt{Pigeons} data sets. We further augment the analysis by comparing all curves from the \texttt{Athens} data set and from the adversarial \texttt{OV} data set. 
To broaden our evaluation, we also include the remaining real-world data sets, where, due to time constraints, we subsample curves and compare all pairs among them. The experiments are primarily bottlenecked by the \texttt{Unleashed} implementation\footnote{A time limit of $5$ minutes was hit multiple times}, and the full evaluation requires several days to complete. Our results indicate no substantial qualitative differences between sampling $100 \times 100$ or $300 \times 300$ curves from the same data.

\newcommand{\ra}[1]{\renewcommand{\arraystretch}{#1}}
\begin{table}[h]
    \centering
    \ra{1.1}
    \begin{tabular}{@{}llrrrl@{}}
    \toprule
      \textbf{Name} & \textbf{Real world} & \textbf{\# curves} & 
      \shortstack{\textbf{total} \\  \textbf{\# vertices}} & \shortstack{\textbf{median} \\  \textbf{\# vertices}} & \shortstack{\textbf{\# of curves} \\  \textbf{compared }} \\
    \midrule
      \texttt{Geolife}  & yes & 18671    & 24895648      & 506  & 300 choose 2\\
      \texttt{Pigeons}  & yes & 539    & 6229319      & 7502  & 300 choose 2\\
      \texttt{Athens}   & yes & 120    & 72439      & 662 & 120 choose 2\\
                OV     & no & 220 & 39060 & 130 &  110 (already pairs)\\
      \texttt{Chicago}  & yes & 888   & 118360  & 121 & 100 choose 2\\
      \texttt{Berlin}   & yes & 27188 & 192223  & 7 & 100 choose 2 \\
      \texttt{Drifter}  & yes & 2012  & 1792084 & 997 & 100 choose 2 \\
      \texttt{UniD}     & yes & 362   & 214076  & 495 & 100 choose 2 \\

    \bottomrule
    \end{tabular}
    \caption{Overview of all data sets used in our experiments.}
    \label{tab:data_sets}
\end{table}

\begin{table}[h]
    \centering
    \ra{1.1}
    \begin{tabular}{@{}llrrrl@{}}
    \toprule
      \textbf{Name} & \textbf{Real world} & \textbf{\# curves} & 
      \shortstack{\textbf{total} \\  \textbf{\# vertices}} & \shortstack{\textbf{median} \\  \textbf{\# vertices}} & \shortstack{\textbf{\# of curves} \\  \textbf{compared }} \\
    \midrule
      \texttt{Geolife}  & yes & 300    & 396676      & 729  & 300 choose 2\\
      \texttt{Pigeons}  & yes & 300    & 2805495      & 5411  & 300 choose 2\\
      \texttt{Athens}   & yes & 120    & 72439      & 662 & 120 choose 2\\
                OV     & no & 220 & 39060 & 130 &  110 (already pairs) \\
      \texttt{Chicago}  & yes & 100   & 12731  & 119 & 100 choose 2 \\
      \texttt{Berlin}   & yes & 100 & 723  & 7 & 100 choose 2 \\
      \texttt{Drifter}  & yes & 100  & 78080 & 804 & 100 choose 2 \\
      \texttt{UniD}     & yes & 100   & 55572  & 421 & 100 choose 2 \\

    \bottomrule
    \end{tabular}
    \caption{Overview of all slices of data sets used in our experiments.}
    \label{tab:data_sets_slice}
\end{table}

\begin{table}[h!]

\centering
\small 
\setlength{\tabcolsep}{3pt}
\begin{tabular}{@{}l|r|r|r@{}}
\toprule
\textbf{Data set} & \textbf{Max} & \textbf{Mean} & \textbf{Median}\\
\midrule

\multicolumn{4}{c}{\textbf{\texttt{Exact}}} \\[0.5em]
\texttt{OV} & 0 & 0 & 0\\
\texttt{athens} & 0 & 0 & 0\\
\texttt{UnID} & 0 & 0 & 0\\
\bottomrule
\multicolumn{4}{c}{\textbf{\texttt{BKN}}} \\[0.5em]
\texttt{OV} & ${3.039939}\cdot 10^{-13}$ & ${2.756980}\cdot 10^{-13}$ & ${3.039939}\cdot 10^{-13}$\\
\texttt{athens} & ${9.337306}\cdot 10^{-10}$ & ${5.865283}\cdot 10^{-12}$ & ${1.269745}\cdot 10^{-12}$\\
\texttt{UnID} & ${9.337306}\cdot 10^{-10}$ & ${5.865283}\cdot 10^{-12}$ & ${1.269745}\cdot 10^{-12}$\\
\bottomrule

\multicolumn{4}{c}{\textbf{\texttt{Unleashed}}} \\[0.5em]
\texttt{OV} & ${4.499666}\cdot 10^{-6}$ & ${1.430129}\cdot 10^{-6}$ & ${1.181895}\cdot 10^{-16}$\\
\texttt{athens} & ${9.267924}\cdot 10^{-6}$ & ${1.948025}\cdot 10^{-7}$ & ${9.974841}\cdot 10^{-15}$\\
\texttt{UnID} & ${9.337306}\cdot 10^{-10}$ & ${5.865283}\cdot 10^{-12}$ & ${1.269745}\cdot 10^{-12}$\\
\bottomrule

\end{tabular}
\caption{Accuracy errors, obtained by converting exact square root distances to a \texttt{double}.}
\label{tab:accuracy}
\end{table}

\subparagraph{Criteria.}
Our primary evaluation metric is running time (including file-reading time). For exactness, our approach is the most accurate if the input coordinates fit within a $32$-bit integer, followed closely by \texttt{BKN}; Table~\ref{tab:accuracy} briefly reports accuracies for the data sets satisfying this condition. We could compare across all data by truncating all input to 32-bit coordinates, but we refrain from doing so since \texttt{BKN} can handle this data natively (albeit inexact).

\newpage
\subsection{Comparing frog-based approaches.}

We first compare our frog-based implementations to \texttt{Unleashed}.  
This comparison is motivated mainly by scientific curiosity: exact-value computation is more difficult than computing an approximation, so one naturally expects a performance penalty.  
Our interest is in quantifying this penalty—in other words, in understanding the practical cost of exactness.  
Table~\ref{tab:stats_frogs} compares the efficacy of the $\texttt{Unleashed}$ implementation versus \texttt{Exact\_Sweepline} (the comparison remains roughly the same when we set for $\texttt{Unleashed}$ the boolean \texttt{accept\_approximation=true}, or when we compare to \texttt{Exact\_Dijkstra} instead).

As expected, the median running time of our exact variants deteriorates.  
However, we note that our exact approach has considerably better worst-case performance, even on this advantageous real-world data. On some instances, \texttt{Unleashed} is a factor 100 slower. Whenever our worst-case is slower, it remains within competitive distance. 
For data sets where curves have sizable complexity, our method becomes the preferred method from the 90th percentile onwards. 
On the mapconstruction data sets, the median number of vertices is very small. The  OV data set is a difficult instance where our exact computations are extremely costly. We note that except for worst-case outliers, the approaches are relatively competitive. 
We think that this data nicely represent the cost of exact value computation, and the benefit of worst-case guarantees for runtime.

\begin{table}[h!]
\centering
\caption{Summary of running times in seconds. For both \texttt{Unleashed} (\texttt{accept\_approx = false}) and \texttt{Exact\_Sweeline}. The best times across the two methods are in blue.}
\label{tab:stats_frogs}
\small 
\setlength{\tabcolsep}{3pt}
\begin{tabular}{@{}l|r|r|r|r|r|r|r|r@{}}
\toprule
\textbf{Data set} & \textbf{Max} & \textbf{Mean} & \textbf{StdDev} & \textbf{Median} & \textbf{\texttt{80th \%ile}} & \textbf{\texttt{90th \%ile}} & \textbf{\texttt{95th \%ile}} & \textbf{\texttt{99th \%ile}} \\
\midrule
\multicolumn{9}{c}{\textbf{\texttt{Unleashed (accept\_approx=false)}}} \\[0.5em]
\texttt{geolife} & 233.6160 & 0.3673 & 3.7389 & \textcolor{blue!70!black}{\textbf{0.0069}} & \textcolor{blue!70!black}{\textbf{0.0621}} & 0.3263 & 1.0108 & 6.4489 \\
\texttt{pigeons} & >300 & 0.8508 & 9.0841 & \textcolor{blue!70!black}{\textbf{0.0222}} & \textcolor{blue!70!black}{\textbf{0.1645}} & 0.9395 & 3.7772 & 13.5008 \\
\texttt{athens} & 29.1741 & 0.2133 & 1.3476 & \textcolor{blue!70!black}{\textbf{0.0190}} & \textcolor{blue!70!black}{\textbf{0.1091}} & \textcolor{blue!70!black}{\textbf{0.2808}} & \textcolor{blue!70!black}{\textbf{0.5125}} & 3.6473 \\
\texttt{OV} & \textcolor{blue!70!black}{\textbf{24.6879}} & \textcolor{blue!70!black}{\textbf{1.2132}} & 3.2387 & \textcolor{blue!70!black}{\textbf{0.0580}} & \textcolor{blue!70!black}{\textbf{1.1886}} & \textcolor{blue!70!black}{\textbf{3.7313}} & \textcolor{blue!70!black}{\textbf{5.2031}} & \textcolor{blue!70!black}{\textbf{15.0121}} \\
\texttt{chicago} & 0.4679 & \textcolor{blue!70!black}{\textbf{0.0036}} & 0.0147 & \textcolor{blue!70!black}{\textbf{0.0024}} & \textcolor{blue!70!black}{\textbf{0.0040}} & \textcolor{blue!70!black}{\textbf{0.0053}} & \textcolor{blue!70!black}{\textbf{0.0063}} & \textcolor{blue!70!black}{\textbf{0.0088}} \\
\texttt{berlin} & 0.0059 & \textcolor{blue!70!black}{\textbf{0.0001}} & 0.0001 & \textcolor{blue!70!black}{\textbf{0.0001}} & \textcolor{blue!70!black}{\textbf{0.0001}} & \textcolor{blue!70!black}{\textbf{0.0001}} & \textcolor{blue!70!black}{\textbf{0.0001}} & \textcolor{blue!70!black}{\textbf{0.0002}} \\
\texttt{uniD} & \textcolor{blue!70!black}{\textbf{0.1402}} & \textcolor{blue!70!black}{\textbf{0.0023}} & 0.0048 & \textcolor{blue!70!black}{\textbf{0.0019}} & \textcolor{blue!70!black}{\textbf{0.0027}} & \textcolor{blue!70!black}{\textbf{0.0032}} & \textcolor{blue!70!black}{\textbf{0.0041}} & \textcolor{blue!70!black}{\textbf{0.0076}} \\
\texttt{drifter} & 61.2818 & 0.1199 & 1.1781 & \textcolor{blue!70!black}{\textbf{0.0091}} & \textcolor{blue!70!black}{\textbf{0.0555}} & 0.1919 & 0.4474 & 1.9376 \\
\bottomrule
\multicolumn{9}{c}{\textbf{\texttt{Exact (with\_dijkstra=false)}}} \\[0.5em]
\texttt{geolife} & \textcolor{blue!70!black}{\textbf{17.5543}} & \textcolor{blue!70!black}{\textbf{0.0984}} & 0.1903 & 0.0505 & 0.1330 & \textcolor{blue!70!black}{\textbf{0.2312}} & \textcolor{blue!70!black}{\textbf{0.3468}} & \textcolor{blue!70!black}{\textbf{0.7444}} \\
\texttt{pigeons} & \textcolor{blue!70!black}{\textbf{2.9708}} & \textcolor{blue!70!black}{\textbf{0.1649}} & 0.1472 & 0.1265 & 0.2305 & \textcolor{blue!70!black}{\textbf{0.3153}} & \textcolor{blue!70!black}{\textbf{0.4130}} & \textcolor{blue!70!black}{\textbf{0.7262}} \\
\texttt{athens} & \textcolor{blue!70!black}{\textbf{2.6424}} & \textcolor{blue!70!black}{\textbf{0.1617}} & 0.2359 & 0.0851 & 0.2136 & 0.3816 & 0.5906 & \textcolor{blue!70!black}{\textbf{1.2356}} \\
\texttt{OV} & 28.2322 & 4.5376 & 6.0343 & 1.5410 & 8.5852 & 13.5812 & 16.9130 & 23.8446 \\
\texttt{chicago} & \textcolor{blue!70!black}{\textbf{0.0522}} & 0.0075 & 0.0073 & 0.0043 & 0.0124 & 0.0179 & 0.0231 & 0.0326 \\
\texttt{berlin} & \textcolor{blue!70!black}{\textbf{0.0029}} & 0.0004 & 0.0002 & 0.0003 & 0.0004 & 0.0005 & 0.0007 & 0.0014 \\
\texttt{uniD} & 0.2809 & 0.0077 & 0.0102 & 0.0064 & 0.0096 & 0.0114 & 0.0139 & 0.0266 \\
\texttt{drifter} & \textcolor{blue!70!black}{\textbf{1.6604}} & \textcolor{blue!70!black}{\textbf{0.0743}} & 0.0889 & 0.0505 & 0.1123 & \textcolor{blue!70!black}{\textbf{0.1609}} & \textcolor{blue!70!black}{\textbf{0.2215}} & \textcolor{blue!70!black}{\textbf{0.4069}} \\
\bottomrule
\end{tabular}
\end{table}

\subsection{Comparing Fréchet distance computations.}

We now provide a  comparison between our exact frog-based method, the approach of~\cite{harpeled_et_al:LIPIcs.SoCG.2025.54}, and the state-of-the-art \texttt{BKN} implementation~\cite{bringmann2019Walking}.  
In this wider evaluation, we consider two configurations of \texttt{Unleashed}: the version with \texttt{accept\_approximation=false} (which, still returns an approximation) and the supposedly faster configuration that sets \texttt{accept\_approximation=true}.
Our exact implementations use either Dijkstra or a sweepline.

\begin{table}[h!]

\centering
\caption{Running time statistics (s). The best running time in bold, the second best in blue.}
\label{tab:stats_all}
\small 
\setlength{\tabcolsep}{3pt}
\begin{tabular}{@{}l|r|r|r|r|r|r|r|r@{}}
\toprule
\textbf{Data set} & \textbf{Max} & \textbf{Mean} & \textbf{StdDev} & \textbf{Median} 
& \textbf{\texttt{80th \%ile}} & \textbf{\texttt{90th \%ile}} 
& \textbf{\texttt{95th \%ile}} & \textbf{\texttt{99th \%ile}} \\
\midrule
\multicolumn{9}{c}{\textbf{\texttt{BKN}}} \\[0.5em]
\texttt{geolife}    & \textbf{0.0767} & \textbf{0.0024} & 0.0016 & \textbf{0.0021} & \textbf{0.0027} & \textbf{0.0034} & \textbf{0.0042} & \textbf{0.0070} \\
\texttt{pigeons}   & \textbf{0.1143} & \textbf{0.0040} & 0.0022 & \textbf{0.0035} & \textbf{0.0048} & \textbf{0.0066} & \textbf{0.0076} & \textbf{0.0102} \\
\texttt{athens}    & \textbf{0.0075} & \textbf{0.0027} & 0.0008 & \textbf{0.0025} & \textbf{0.0033} & \textbf{0.0038} & \textbf{0.0043} & \textbf{0.0055} \\
\texttt{OV}    & \textbf{0.0728} & \textbf{0.0208} & 0.0206 & \textbf{0.0111} & \textbf{0.0376} & \textbf{0.0580} & \textbf{0.0629} & \textbf{0.0682} \\
\texttt{chicago} & \textbf{0.0047} & \textbf{0.0018} & 0.0002 & \textbf{0.0018} & \textbf{0.0019} & \textbf{0.0020} & \textbf{0.0021} & \textbf{0.0024} \\
\texttt{berlin} & \textcolor{blue!70!black}{\textbf{0.0032}} & 0.0016 & 0.0002 & 0.0015 & 0.0017 & 0.0018 & 0.0018 & 0.0019 \\
\texttt{uniD} & \textbf{0.0046} & \textbf{0.0018} & 0.0002 & \textbf{0.0018} & \textbf{0.0020} & \textbf{0.0021} & \textbf{0.0022} & \textbf{0.0024} \\
\texttt{drifter}  & \textbf{0.0144} & \textbf{0.0024} & 0.0009 & \textbf{0.0022} & \textbf{0.0027} & \textbf{0.0032} & \textbf{0.0037} & \textbf{0.0060} \\

\bottomrule

\multicolumn{9}{c}{\textbf{\texttt{Unleashed (accept\_approx=true)}}} \\[0.5em]
\texttt{geolife} & 250.5627 & 0.3367 & 3.5109 & 0.0074 & 0.0623 & 0.3244 & 0.9207 & 5.8376 \\
\texttt{pigeons}    & 274.5574 & 0.7529 & 4.6066 & 0.0225 & \textcolor{blue!70!black}{\textbf{0.1623}} & 0.9021 & 3.7495 & 13.2745 \\
\texttt{athens}    & 29.4879 & 0.1754 & 1.1123 & 0.0195 & \textcolor{blue!70!black}{\textbf{0.1081}} & 0.2415 & 0.4817 & 2.4792 \\
\texttt{OV}    & 24.8995 & 1.2463 & 3.2985 & 0.0609 & 1.3038 & 3.8347 & 5.5689 & 15.1980 \\
\texttt{chicago} & 0.5972 & \textcolor{blue!70!black}{\textbf{0.0035}} & 0.0144 & \textcolor{blue!70!black}{\textbf{0.0023}} & \textcolor{blue!70!black}{\textbf{0.0040}} & 0.0054 & 0.0064 & 0.0093 \\
\texttt{berlin} & 0.1214 & \textbf{0.0001} & 0.0018 & \textbf{0.0001} & \textbf{0.0001} & \textbf{0.0001} & \textcolor{blue!70!black}{\textbf{0.0002}} & \textbf{0.0002} \\
\texttt{uniD} & \textcolor{blue!70!black}{\textbf{0.1239}} & 0.0024 & 0.0045 & 0.0020 & 0.0029 & 0.0035 & 0.0042 & 0.0086 \\
\texttt{drifter} & 60.7184 & 0.1177 & 1.1667 & 0.0095 & \textcolor{blue!70!black}{\textbf{0.0529}} & 0.1730 & 0.4747 & 1.8644 \\
\bottomrule

\multicolumn{9}{c}{\textbf{\texttt{Unleashed (accept\_approx=false)}}} \\[0.5em]
\texttt{geolife} & 233.6160 & 0.3673 & 3.7389 & \textcolor{blue!70!black}{\textbf{0.0069}} & \textcolor{blue!70!black}{\textbf{0.0621}} & 0.3263 & 1.0108 & 6.4489 \\
\texttt{pigeons}    & >300 & 0.8508 & 9.0841 & \textcolor{blue!70!black}{\textbf{0.0222}} & 0.1645 & 0.9395 & 3.7772 & 13.5008 \\
\texttt{athens}    & 29.1741 & 0.2133 & 1.3476 & \textcolor{blue!70!black}{\textbf{0.0190}} & 0.1091 & 0.2808 & 0.5125 & 3.6473 \\
\texttt{OV}    & \textcolor{blue!70!black}{\textbf{24.6879}} & \textcolor{blue!70!black}{\textbf{1.2132}} & 3.2387 & \textcolor{blue!70!black}{\textbf{0.0580}} & \textcolor{blue!70!black}{\textbf{1.1886}} & \textcolor{blue!70!black}{\textbf{3.7313}} & \textcolor{blue!70!black}{\textbf{5.2031}} & \textcolor{blue!70!black}{\textbf{15.0121}} \\
\texttt{chicago} & 0.4679 & 0.0036 & 0.0147 & 0.0024 & 0.0040 & \textcolor{blue!70!black}{\textbf{0.0053}} & \textcolor{blue!70!black}{\textbf{0.0063}} & \textcolor{blue!70!black}{\textbf{0.0088}} \\
\texttt{berlin} & 0.0059 & \textbf{0.0001} & 0.0001 & \textbf{0.0001} & \textbf{0.0001} & \textbf{0.0001} & \textbf{0.0001} & \textbf{0.0002} \\
\texttt{uniD} & 0.1402 & \textcolor{blue!70!black}{\textbf{0.0023}} & 0.0048 & \textcolor{blue!70!black}{\textbf{0.0019}} & \textcolor{blue!70!black}{\textbf{0.0027}} & \textcolor{blue!70!black}{\textbf{0.0032}} & \textcolor{blue!70!black}{\textbf{0.0041}} & \textcolor{blue!70!black}{\textbf{0.0076}} \\
\texttt{drifter} & 61.2818 & 0.1199 & 1.1781 & \textcolor{blue!70!black}{\textbf{0.0091}} & 0.0555 & 0.1919 & 0.4474 & 1.9376 \\
\bottomrule

\multicolumn{9}{c}{\textbf{\texttt{Exact (with\_dijkstra=true)}}} \\[0.5em]
\texttt{geolife} & \textcolor{blue!70!black}{\textbf{7.6259}} & 0.1028 & 0.1874 & 0.0525 & 0.1381 & 0.2414 & 0.3633 & 0.7742 \\
\texttt{pigeons}    & 3.1337 & 0.1672 & 0.1563 & 0.1268 & 0.2322 & 0.3201 & 0.4239 & 0.7432 \\
\texttt{athens}    & 5.3387 & 0.1975 & 0.3426 & 0.0912 & 0.2434 & 0.4529 & 0.7634 & 1.6910 \\
\texttt{OV}    & 24.4039 & 3.3321 & 5.0168 & 0.7700 & 6.7386 & 10.2982 & 14.0325 & 19.9612 \\
\texttt{chicago}  & 0.0543 & 0.0076 & 0.0076 & 0.0043 & 0.0122 & 0.0182 & 0.0241 & 0.0344 \\
\texttt{berlin} & 0.0037 & \textcolor{blue!70!black}{\textbf{0.0004}} & 0.0002 & \textcolor{blue!70!black}{\textbf{0.0003}} & \textcolor{blue!70!black}{\textbf{0.0004}} & \textcolor{blue!70!black}{\textbf{0.0004}} & 0.0006 & \textcolor{blue!70!black}{\textbf{0.0014}} \\
\texttt{uniD} & 0.2445 & 0.0077 & 0.0098 & 0.0064 & 0.0095 & 0.0114 & 0.0138 & 0.0271 \\
\texttt{drifter} & 1.7855 & 0.0806 & 0.1061 & 0.0519 & 0.1182 & 0.1748 & 0.2428 & 0.4936 \\
\bottomrule

\multicolumn{9}{c}{\textbf{\texttt{Exact (with\_dijkstra=false)}}} \\[0.5em]
\texttt{geolife} & 17.5543 & \textcolor{blue!70!black}{\textbf{0.0984}} & 0.1903 & 0.0505 & 0.1330 & \textcolor{blue!70!black}{\textbf{0.2312}} & \textcolor{blue!70!black}{\textbf{0.3468}} & \textcolor{blue!70!black}{\textbf{0.7444}} \\
\texttt{pigeons}    & \textcolor{blue!70!black}{\textbf{2.9708}} & \textcolor{blue!70!black}{\textbf{0.1649}} & 0.1472 & 0.1265 & 0.2305 & \textcolor{blue!70!black}{\textbf{0.3153}} & \textcolor{blue!70!black}{\textbf{0.4130}} & \textcolor{blue!70!black}{\textbf{0.7262}} \\
\texttt{athens}    & \textcolor{blue!70!black}{\textbf{2.6424}} & \textcolor{blue!70!black}{\textbf{0.1617}} & 0.2359 & 0.0851 & 0.2136 & \textcolor{blue!70!black}{\textbf{0.3816}} & \textcolor{blue!70!black}{\textbf{0.5906}} & \textcolor{blue!70!black}{\textbf{1.2356}} \\
\texttt{OV}    & 28.2322 & 4.5376 & 6.0343 & 1.5410 & 8.5852 & 13.5812 & 16.9130 & 23.8446 \\
\texttt{chicago} & \textcolor{blue!70!black}{\textbf{0.0522}} & 0.0075 & 0.0073 & 0.0043 & 0.0124 & 0.0179 & 0.0231 & 0.0326 \\
\texttt{berlin} & \textbf{0.0029} & \textcolor{blue!70!black}{\textbf{0.0004}} & 0.0002 & \textcolor{blue!70!black}{\textbf{0.0003}} & \textcolor{blue!70!black}{\textbf{0.0004}} & 0.0005 & 0.0007 & \textcolor{blue!70!black}{\textbf{0.0014}} \\
\texttt{uniD} & 0.2809 & 0.0077 & 0.0102 & 0.0064 & 0.0096 & 0.0114 & 0.0139 & 0.0266 \\
\texttt{drifter} & \textcolor{blue!70!black}{\textbf{1.6604}} & \textcolor{blue!70!black}{\textbf{0.0743}} & 0.0889 & 0.0505 & 0.1123 & \textcolor{blue!70!black}{\textbf{0.1609}} & \textcolor{blue!70!black}{\textbf{0.2215}} & \textcolor{blue!70!black}{\textbf{0.4069}} \\
\bottomrule
\end{tabular}
\end{table}

\subsection{Results}

Table~\ref{tab:stats_all} shows the full set of experimental results.
These plots highlight the volatility of the frog-based approach: the standard deviation of these approaches is considerably larger than \texttt{BKN} (although lessened for our implementation). 
Regardless of evaluation criterion (maximum, mean or median running time) the \texttt{BKN} algorithm dominates in performance\footnote{Due to technical constraints, there is system overhead influencing the measured time for \texttt{BKN}. This is particularly noticeable for instances of the \texttt{berlin} data set since they have so few vertices.}. 

We illustrate the performance profile of these algorithms using Dolan--Moré plots (Figures~\ref{fig:dolan_geolife}--\ref{fig:dolan_OV}). 
These plots can be read as follows.  
The $x$-axis shows a factor $\tau \ge 1$, indicating how many times slower an algorithm may be compared to the best on each instance.  
The $y$-axis shows the percentage of instances on which this holds. 
Thus, if an algorithm's plot contains the point $(p, \tau)$ then on $p$ percent of the instances, it is $\tau$ times slower than the best performing algorithm. 
Similarly, $\tau=1$ reveals how often each algorithm is the best, while a horizontal line such as $p=90$ indicates how much slower the worst 10\% of cases become (which, for \texttt{BKN}, is minimal).  
As $\tau$ increases, the curves show how quickly each algorithm “catches up’’.
An algorithm is considered better if its curve approaches the horizontal line $p = 100$ more rapidly.  
These profiles mirror our observations: \texttt{BKN} is consistent; our method lags on median instances, however if the data set contains high enough complexity curves, it frequently and significantly overtakes \texttt{Unleashed} as we approach the worst case.

\begin{figure}[H]
    \centering
    \includegraphics[width=0.49\textwidth]{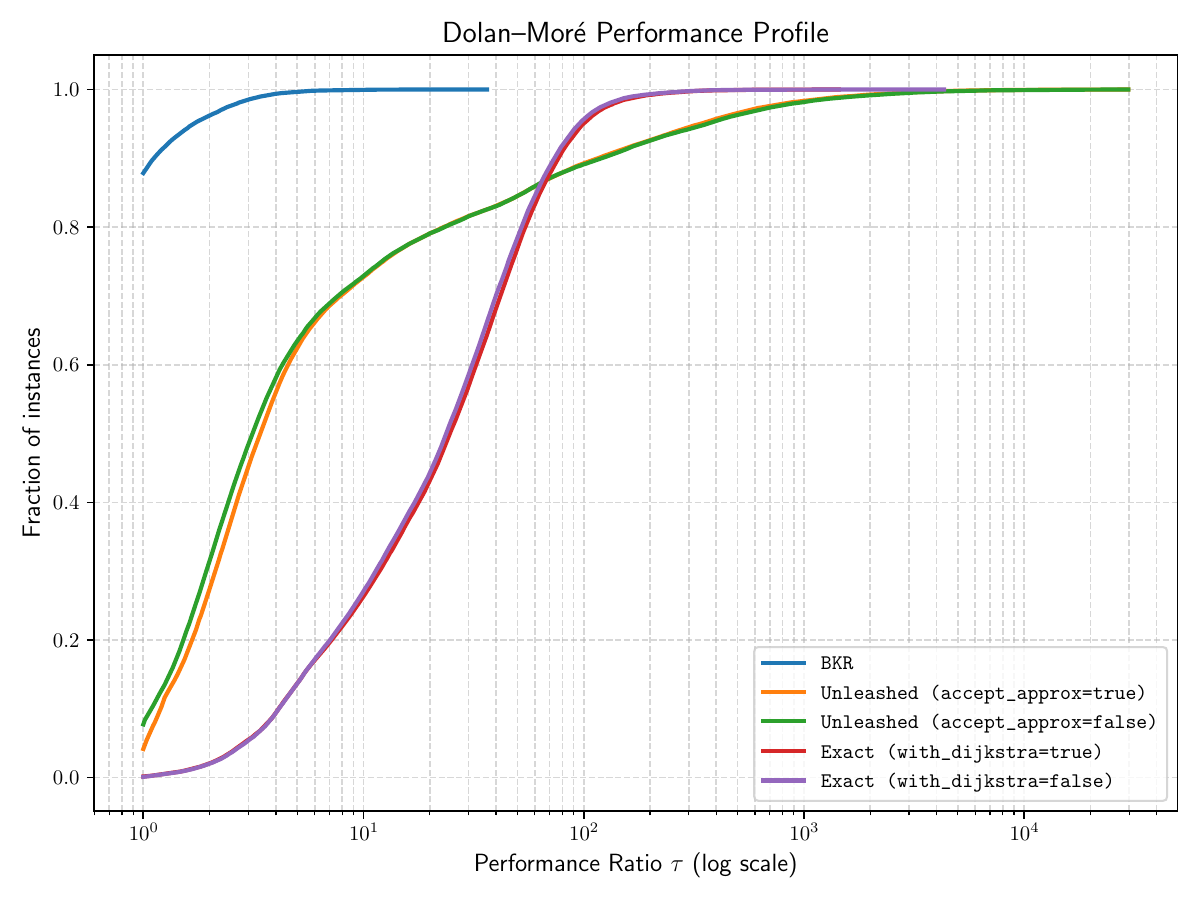}
    \includegraphics[width=0.49\textwidth]{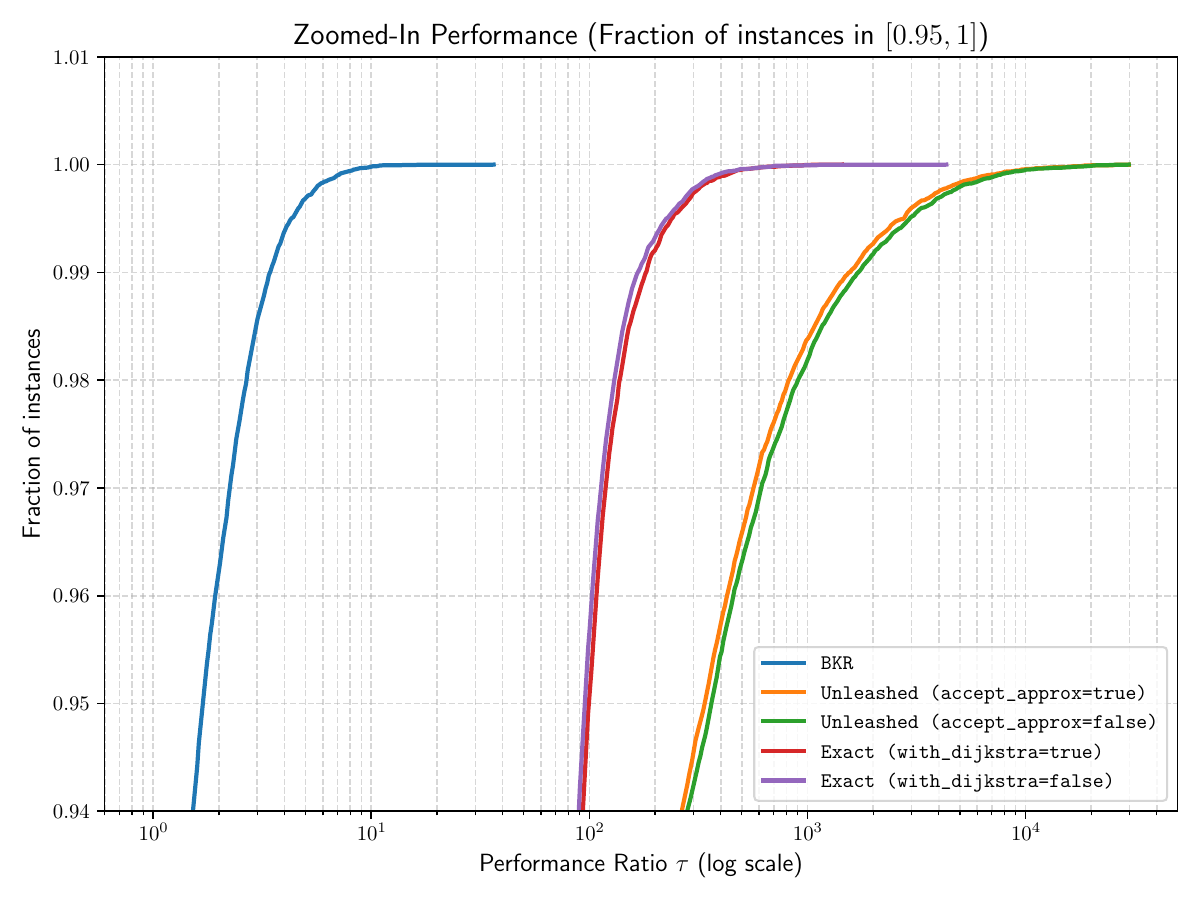}
    \caption{Dolan--Moré performance profile on all pairs of the first 300 curves from the \texttt{geolife} data set (44850 in total) from \cite{harpeled_et_al:LIPIcs.SoCG.2025.54}.}
    \label{fig:dolan_geolife}
\end{figure}

\begin{figure}[H]
    \centering
    \includegraphics[width=0.49\textwidth]{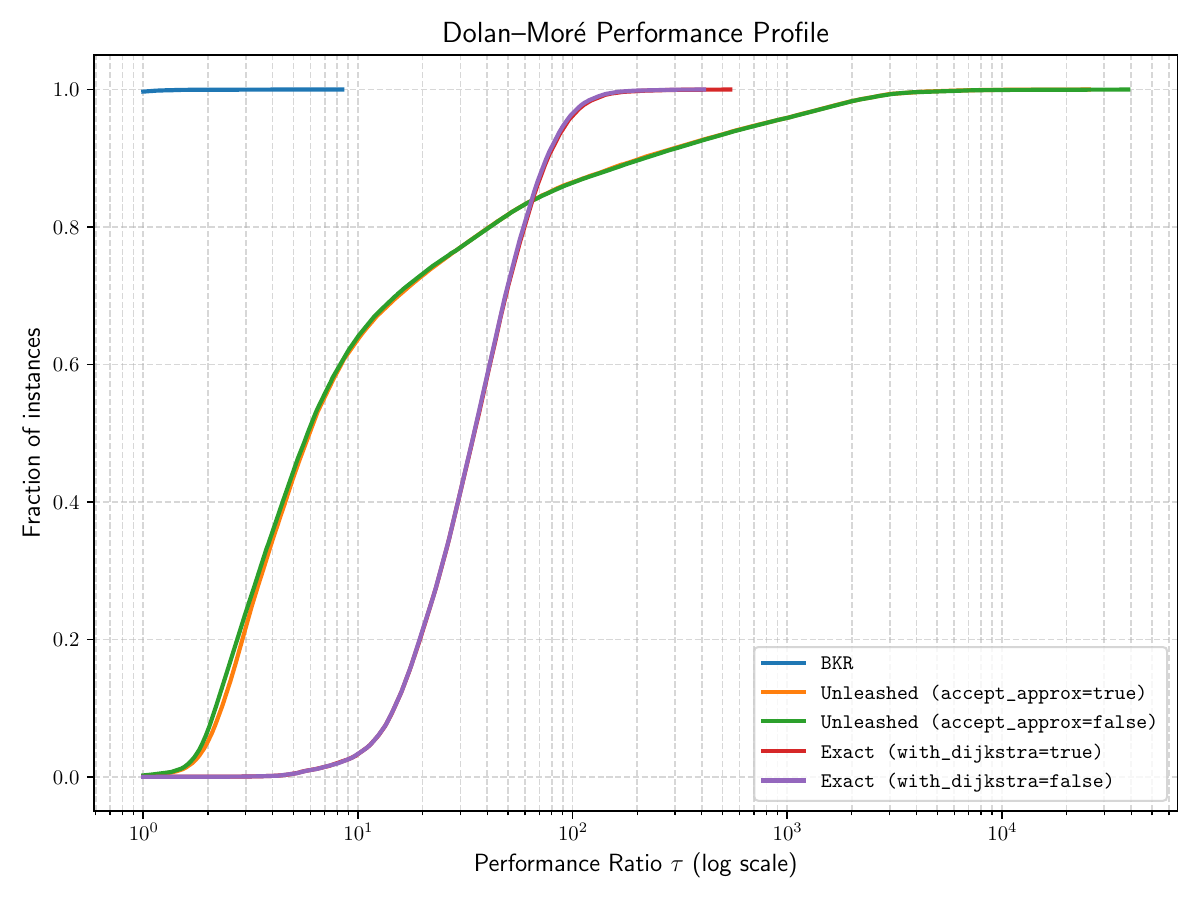}
    \includegraphics[width=0.49\textwidth]{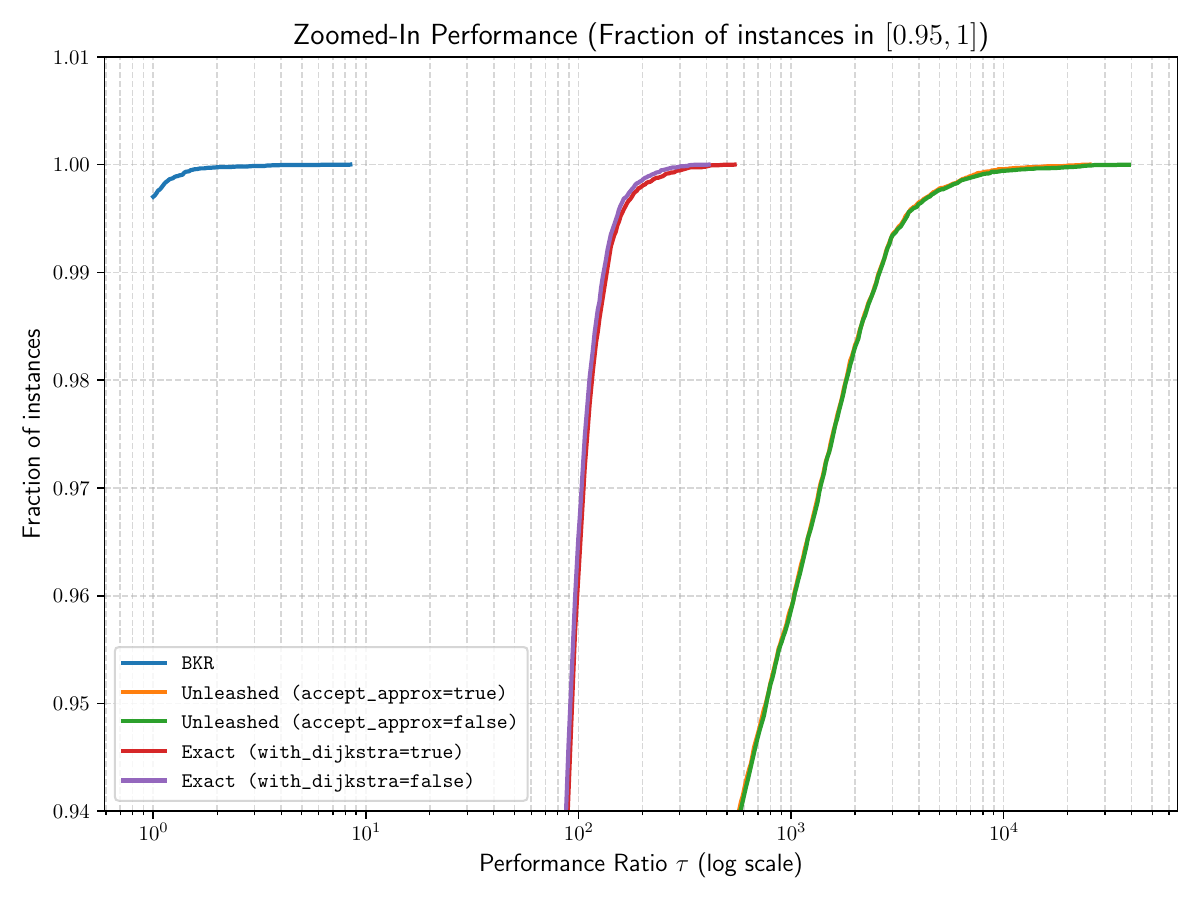}
    \caption{Dolan--Moré performance profile on all pairs of the first 300 curves from the \texttt{pigeons} data set (44850 in total) from \cite{harpeled_et_al:LIPIcs.SoCG.2025.54}.}
    \label{fig:dolan_pigeons}
\end{figure}

\begin{figure}[H]
    \centering
    \includegraphics[width=0.49\textwidth]{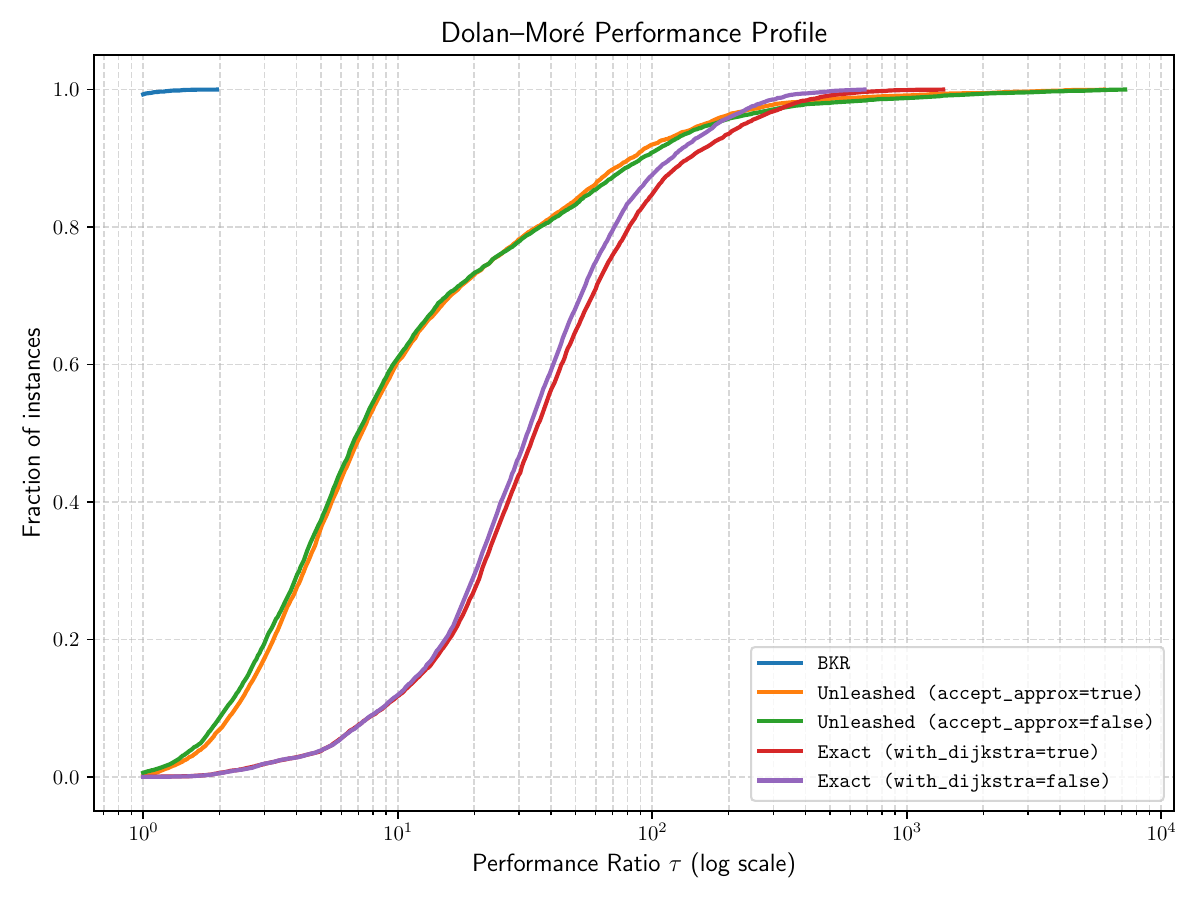}
    \includegraphics[width=0.49\textwidth]{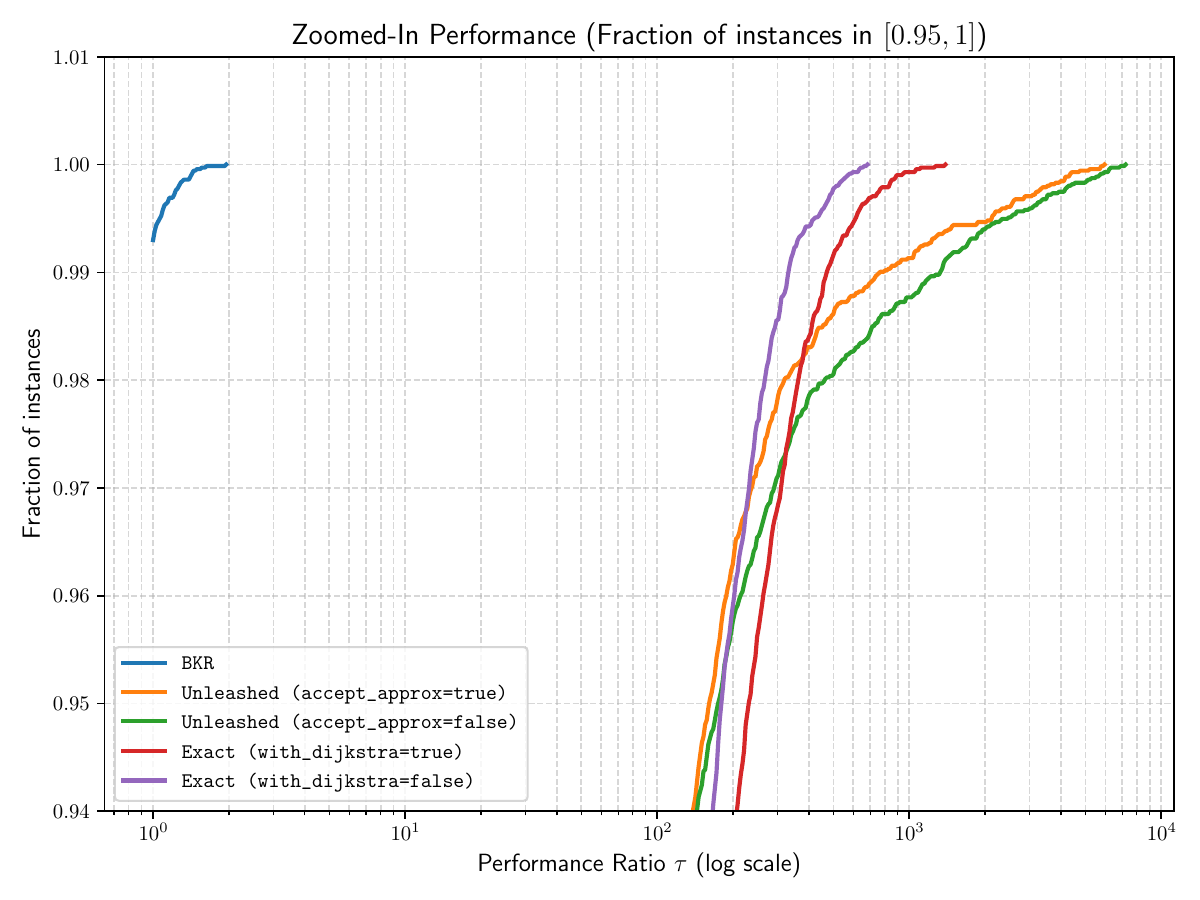}
    \caption{Dolan--Moré profile on pairs from the \texttt{athens} data set (7139 in total) from \cite{vanderhoo2025Efficient}.}
    \label{fig:dolan_athens}
\end{figure}

\begin{figure}[H]
    \centering
    \includegraphics[width=0.49\textwidth]{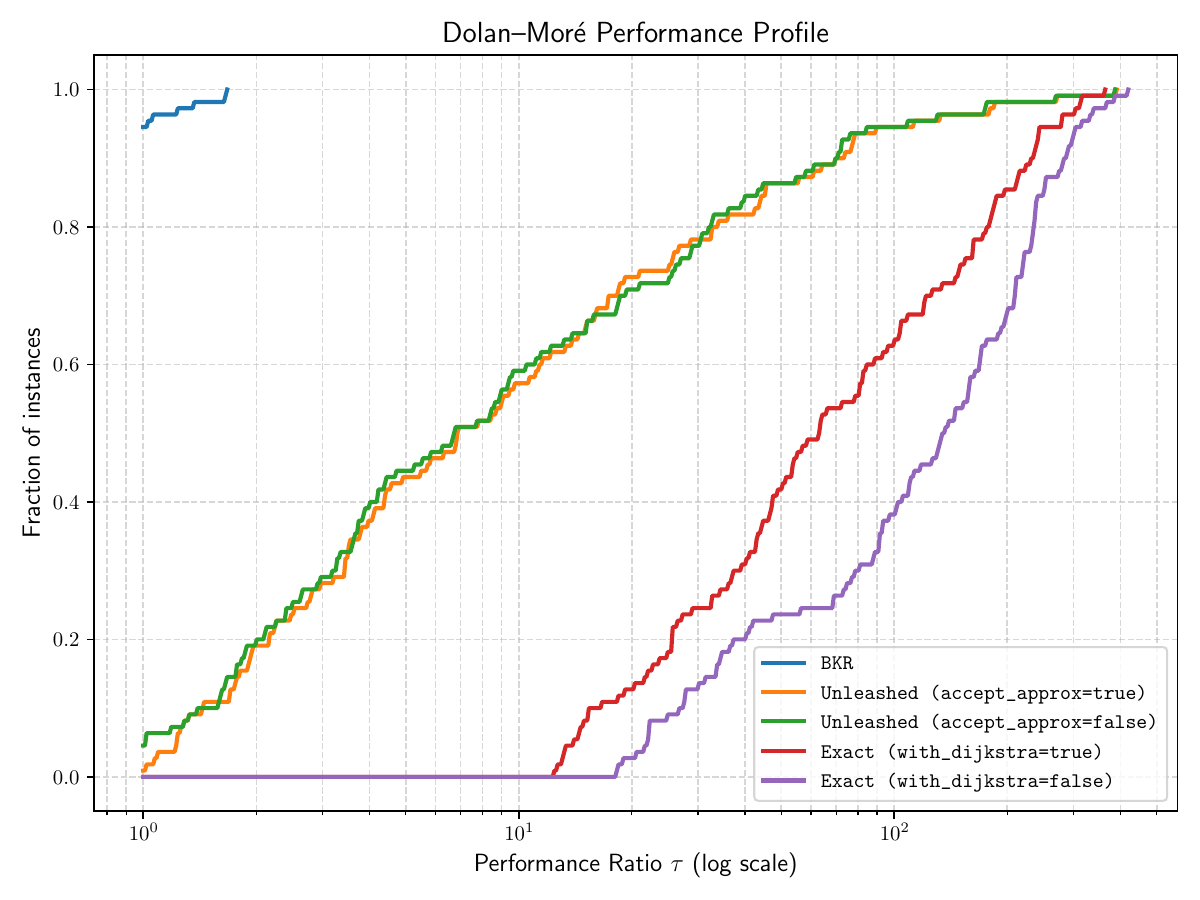}
    \includegraphics[width=0.49\textwidth]{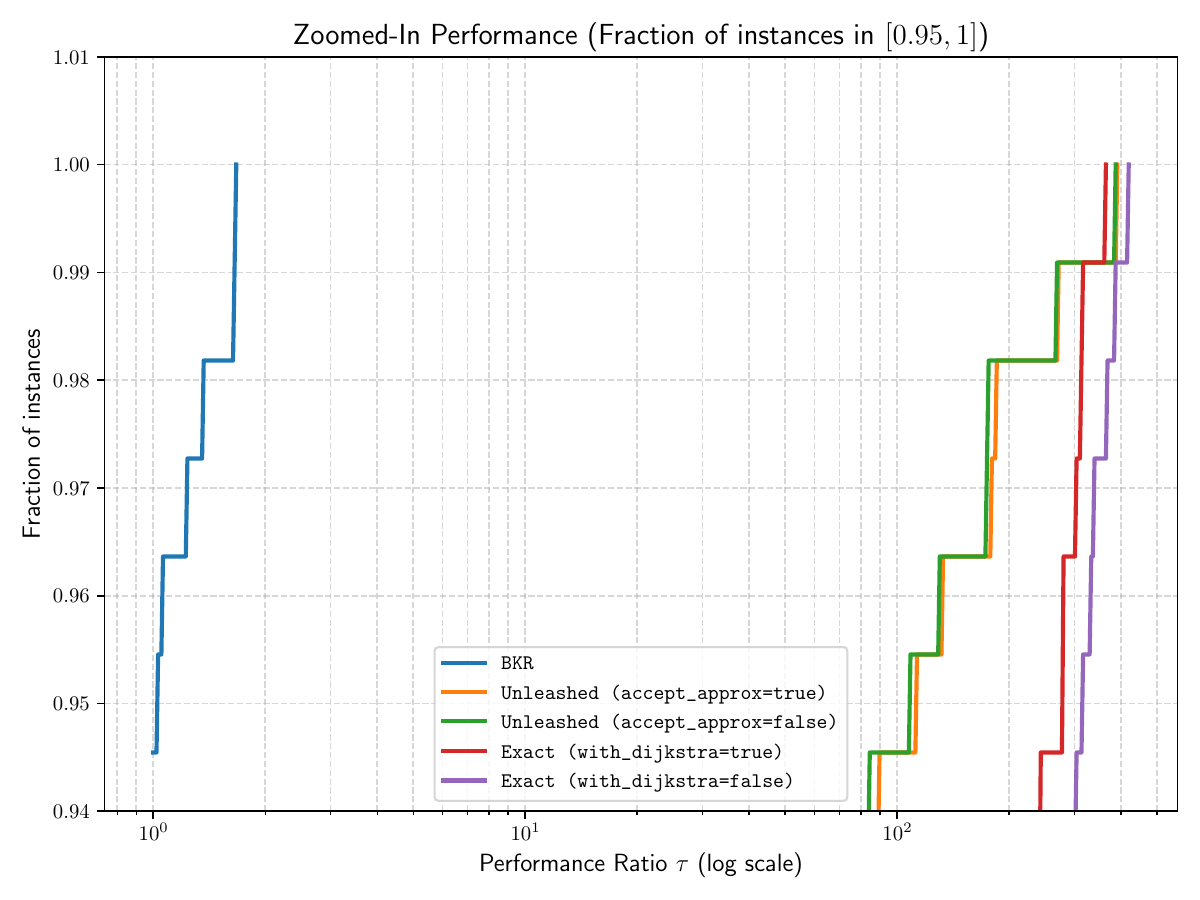}
    \caption{Dolan--Moré performance profile on all pairs from the \texttt{OV} data set (110 in total).}
    \label{fig:dolan_OV}
\end{figure}

\newpage
\section{Conclusion}

We studied the frog-based framework for computing the Fréchet distance introduced in~\cite{harpeled_et_al:LIPIcs.SoCG.2025.54}.  
On the theoretical side, we identified three aspects of the original algorithm that incur avoidable inaccuracies in the distance computation or introduce overhead.  
Specifically, the method of~\cite{harpeled_et_al:LIPIcs.SoCG.2025.54} consists of two main components.  
The first is the computation of the Fréchet distance between $(\pi,\sigma)$ via a discrete weighted graph whose complexity is increased iteratively each round, through a process they call \emph{refinement}. 
Their process heuristically bisects edges of $(\pi,\sigma)$, increasing the graph complexity, until the discrete min-cost path through this graph realises the Fréchet distance. 
We provide a new refinement scheme that offers worst-case guarantees on this process.  
Our more elaborate refinement is careful not to drastically increase the time per round, yet it yields an explicit upper bound on the convergence rate ---  which we consider theoretically meaningful.

The second component is a lossless curve simplification procedure that produces curves $(\overline{\pi},\overline{\sigma})$ of lower complexity such that $\fd(\pi,\sigma)=\fd(\overline{\pi},\overline{\sigma})$ --- effectively reducing the input size.  
Although correct, the original simplification is overly conservative: it relies on a lower bound that is unnecessarily weak, and it avoids geometric corner cases by aggressively propagating vertex slack across a constant but wide neighbourhood.  
Exploiting the weighted-graph structure already inherent in the frog-based method, we present a more refined construction that attains a higher lower bound.  
Consequently, our simplifications have lower complexity.  

Our primary contribution is an empirical evaluation.  
We supply an openly available test suite, which enables a reproducible comparison between frog-based implementations and the highly optimised algorithm of Bringmann, Künnemann and Nusser~\cite{bringmann2019Walking}.  
Our results contrast the findings in~\cite{harpeled_et_al:LIPIcs.SoCG.2025.54}:  
On our broader range of real-world data sets, \texttt{BKN} consistently outperforms both frog-based variants.  
Notably, after removing near-duplicate consecutive points through simplification, the data set curves have \emph{very} few vertices remaining for modern hardware. We have reason to believe that, on larger inputs, the performance gap between our method and \texttt{Unleashed} widens. However, we are not aware of real-world benchmark sets whose median complexity exceeds 1{,}000 ``distinct'' vertices, and random data behaves entirely differently under lossless simplification. We argue that existing benchmarks are becoming outdated and that the community should develop more modern ones.

At first glance, our findings might seem discouraging: our implementation and analysis describe an approach that does not, in practice, match the performance of the current state-of-the-art.  
However, we consider this outcome important for the community.  
Our study offers a transparent and verifiable assessment of the frog-based method and, contrary to earlier claims, demonstrates that its practical behaviour is volatile.

At the same time, our analysis reaffirms that the frog-based framework is theoretically compelling (if you ask us, particularly in its exact variant).
We show that this method in particular enables a significantly more fine-grained lossless simplification procedure than general-purpose methods.  
Nevertheless, our experiments indicate that its main value lies in its theoretical insights rather than in practical performance.

\bibliography{refs.bib}

\appendix

\section{Generating monotonicity events}
\label{app:monotonicity}

Let $\pi$ and $\sigma$ be two curves whose original complexity had $n$ and $m$ vertices respectively, where we added $M$ monotonicity vertices to them.
Consider their \texttt{VE}-graph and let $P$ be a min-cost path from $(1, 1)$ to $(n, m)$. 
We can identify each row $R$ and column $C$ in which $P$ is non-monotone.
We present a subroutine algorithm, whose input the input is a column $C$ where $P$ is not $x$-monotone, corresponding to an edge $e$ of $\pi$. Let $P$, restricted to $C$, go from $a$ to $b$.
We show how to introduce a minimum-size set of vertices $V$ along $e$ such that, in the resulting denser \texttt{VE}-graph, the minimum cost path from $a$ to $b$ is monotone. 
We show that we can apply this subroutine to row $R$ and column $C$ in $O( (n+m+M) \log^2 (n+m+M))$ total time.

\subsection{Solving the subproblem.}
Given $C$ and $a$ and $b$, we create two new cells by extending a horizontal segment from $a$ and $b$. We discard all cells of $C$ that do not intersect $P$. 
This creates a new set of cells $C_1, \ldots, C_k$ ordered from bottom to top, where the bottom left of $C_1$ equals $a$ and the top right of $C_k$ equals $b$. 
This, in turn,  induce horizontal segments. 
Consider some value $\Delta \geq 0$.
We define the $\Delta$ free space as all points $(x, y)$ for which $d(\pi(x), \sigma(y)) \leq \Delta$. We denote this area as $\dfs$ and the horizontal segments induce intervals  $h_1, \ldots, h_{k+1}$ which is each segment intersected with $\dfs$ (these intervals implicitly grow as we increase $\Delta$). 

For each pair $(i, j)$ with $i < j$, $h_i$ can \emph{$\Delta$-reach} $h_j$ if there exists a monotone path from a point on $h_i$ to a point on $h_j$ inside $\dfs$.
Imagine $\Delta$ increasing from $0$ to $\infty$.
We consider four types of \emph{events}:
\begin{itemize}
    \item A \emph{spawn event} for a horizontal segment $h_i$ is the lowest $\Delta$ such that $h_i \neq \emptyset$.
    \item An \emph{undertake} event for an index pair $(i, j)$ is the lowest value $\Delta$ such that for all $\Delta' \geq \Delta$, the left endpoint of $h_i$ is farther left than the left endpoint of $h_j$. 
    \item A \emph{join} event for an index pair $(i, j)$ is the lowest value $\Delta$ for which the intervals $h_i$ and $h_j$ vertically overlap. 
\end{itemize}

\subparagraph{Storing a loser tree and a heap.}
Imagine increasing $\Delta$ from $0$ to $\infty$. 
We compute the minimum number of join events to realise a monotone path from $h_1$ to $h_{k}$ using a linked–list.
We also a tree that we call a \emph{loser tree} over $h_1$ to $h_k$. 
For each node of the loser tree stores a segment.
Specifically, if a node has children storing segments $h_i$ and $h_j$, we store $h_i$ if in $\dfs$ the left endpoint of $h_i$ is farther right to the left endpoint of $h_j$. We store in our heap the undertake event $(i, j)$ --- which may change the ordering between the left andpoints of $h_i$ and $h_j$. 

We  also store a linked list of buckets. 
The bottom bucket $(h_1, h_2, \ldots h_i)$  contains all segments from $h_1$ up to the maximum index $i$ such that there is a path from $h_1$ to $h_i$ in $\dfs$. 
The remaining items of each list store $h_j$ for $j > i$ consecutively. 
Denote by $\ell$ the rightmost left endpoint of all segments in $(h_1, h_2, \ldots h_i)$ in $\dfs$ and let $h_\ell$ be the corresponding segment; note that we can obtain $\ell$ in $O(\log k)$ time from the loser tree.  We store in our heap the spawn event for $h_{i+1}$ and the join event between $h_\ell$ and $h_{i+1}$.

\subparagraph{The argument.}
Suppose that for the current value of $\Delta$, the bottom bucket is $(h_1, h_2, \ldots h_i)$.
Then $\ell$ lies further right than the right endpoint of $h_{i+1}$. 
We pop the next event from the heap and set $\Delta$ to the corresponding value. 
We do a case analysis of which event is popped.

If the event is a spawn event of $h_{i+1}$, we test whether it becomes possible to go from $h_1$ to $h_{i+1}$ by testing if the point that is $h_{i+1}$ lies right of $\ell$. 
If so, we merge $h_{i+1}$ into the bottom bucket. 
We recompute $\ell$ in $O(\log k)$ time using the loser tree, remove the old join event, and add the join event between the new $h_\ell$ and $h_{i+2}$.

If the event is an undertake event in the loser tree, then we update the loser tree.
Note that these events change at most one node in the loser tree at a time, since at the time of an undertake event the left endpoints of the involved intervals coincide with one another.
Thus, we update the loser tree in $O(\log k)$ time and we find the new value $\ell$ and the corresponding segment $h_\ell$ in $O(\log k)$ time. 
Since for pair $(a, b)$, there can be only one undertake event and that invalidates all undertake events of $b$ (since after an undertake, $h_i$ grows faster than $h_j$), the loser tree generates at most $O(k \log k)$ undertake events and therefore this takes $O(k \log^2 k)$ total time.

If the event is the join event between $h_\ell$ and $h_{i+1}$ then, given that it was not possible to go from $h_1$ to $h_{i+1}$ before, it now becomes possible to go from $h_1$ to $h_{i+1}$. 
For the current updated $\Delta$, the single vertical line where $h_1$ and $h_{i+1}$ coincides corresponds to the monotonicity event between their corresponding points on $\sigma$ and the corresponding edge of $\pi$. We add the corresponding monotonicity vertex to $\pi$.
We then merge $h_{i+1}$ into the bottom bucket. 
We recompute $\ell$ in $O(\log k)$ time using the loser tree. We add to the heap the join event between the new $h_\ell$ and $h_{i+2}$ and the spawn event of $h_{i+1}$. 

As a result, we compute in $O(k \log^2 k)$ time the minimum number of monotonicity vertices required for there to be a monotone path in the \texttt{VE}-graph from $a$ to $b$. Note that the minimality comes from the fact that for any monotonicity event that we encounter this is the minimum value $\Delta$ that prevents further progression for a monotone curve. 
Since each path in the \texttt{VE}-graph has at most $O(n + m + M)$ vertices, applying this subroutine to all columns $C$ and rows $R$ where the path is non-monotone takes $O( (n + m + M) \log^2 (n + m + M))$ time:

\begin{lemma}
    Consider their \texttt{VE}-graph and let $P$ be a min-cost path from $(1, 1)$ to $(n, m)$. 
We can identify each row $R$ and column $C$ in which $P$ is non-monotone.
For each column $C$, where $P$ is not $x$-monotone, let $e$ be the corresponding edge of $\pi$. Let $P$, restricted to $C$, go from $a$ to $b$.

In $O( (n + m + M) \log^2 (n + m + M))$ total time we can, for all $(C, e, (a, b))$,  compute the minimum set of monotonicity vertices $V$ on $e$ that need to be introduced such that in the resulting \texttt{VE}-graph, the minimum cost path from $a$ to $b$ is monotone. 
\end{lemma}

\shootup*
\begin{proof}
    Denote by $(\pi, \sigma)$ the two curves that originally had $n$ and $m$ vertices, and that have $M$ vertices added to them. 
    Each iteration, we can compute the min-cost path from the source to the sink in the \texttt{VE}-graph in $O(n\cdot m + M)$ time using a simple linear scan from left to right over this vertex-weighted \texttt{VE}-graph.
    The resulting distance $\ve(\pi, \sigma)$ is at most $\fd(\pi, \sigma)$ and it is exactly $\fd(\pi, \sigma)$ if there exists a path in the \texttt{VE}-graph realising $\ve(\pi, \sigma)$ that is $xy$-monotone. 
    A linear scan informs us that $\ve(\pi, \sigma) = \Delta$.
    Given $\Delta$, a depth-first search over this graph informs us whether there exists a monotone path realising this distance in the \texttt{VE}-graph.
    Thus, each iteration, we can test in $O(n\cdot m + M)$ time whether our algorithm may terminate. 
    If we do not terminate, we invoke the subroutine in $O((n+m+M) \log (n+m+M))$ time. 
    It remains to show that our algorithm terminates  within $O(n^3+m^3)$ iterations. 

    Consider any pair of cells that share either a column, or a row. there are $O(nm^2 + mn^2)$ such pairs. Any path through the graph can be represented as a sequence of such pairs, such that two consecutive pairs do not share a row or column---essentially the maximal sub-rows and sub-columns traversed by the particular path. If the algorithm computes a path that is non-monotone, then every such pair is scrutinized, and a sequence of break points is introduced, such that these pairs may no longer be traversed non-monotonuously. In particular, any later path that shares a pair of cells in its represetntation must be monotonuous between this pair of cells. As there are only $O(nm^2 + mn^2)$ such pairs, and in every iteration at least one such pair is no longer non-monotonuously traverable, after a total of $O(nm^2 + mn^2)$ iterations the algorithm terminates, concluding the proof.
\end{proof}

\section{Lossless simplification}
\label{sec:lossless}

In this section we discuss details of our simplification, and how we obtain tighter lower bounds.

\subsection{Hierarchical Simplification}

We first describe the particular simplification we are using. We maintain a vertex-restricted simplification of the input, where on demand, we may subdivide any edge of the simplification introducing a vertex of the input in-between the two simplification vertices. Along with it, we maintain a joint traversal of every edge of the simplification with the corresponding subcurve of the input curve. This joint traversal is not necessarily optimal, but instead, we focus on a good, and quickly computable joint traversal. 

To compute the joint traversal, let $p_i$ and $p_j$ be the vertices of $\pi$ defining some edge of $\overline{\pi}$. For every $p_k$ with $i\leq k\leq j$, we compute the value $t_k\in[0,1]$ minimizing $\|p_k-(p_i+t_k(p_j-p_i))\|$, that is the point along the edge of the simplification minimizing the distance to $p_k$. From these, we obtain a matching between the edge of the simplification with the corresponding subcurve matching $p_k$ with the point at $\max_{k'\leq k}t_{k'}$. It is easy to verify, that this is a $2$-approximation of the optimal traversal. 
This traversal can be computed time that is linear in the number of vertices of the corresponding subcurve by a simple linear scan. In particular, this scan immediately yields what we define as a \emph{split vertex}, i.e., a vertex realizing the maximum distance under this traversal.

\subsection{A tighter lower bound}

We next describe, how exactly we obtain our lower bound. For this, we endow every edge of the simplification with a negative weight corresponding to the negative maximum distance attained during the greedily computed joint traversal of the edge with its corresponding subcurve. This is the distance realised by the split vertex of the edge.  We compute the lower bound $LB(\overline{\pi},\overline{\sigma})$ by essentially computing the Fréchet distance respecting the edge weights. 

\begin{definition}[Fréchet distance with additive weights]
    Let $\pi$ and $\sigma$ be two polygonal curves, where every edge of $\pi$ and $\sigma$ is endowed with a weight. Let $w_\pi:[0,1]\rightarrow\mathbb{R}$ be the function mapping any $t$ to $0$, if it coincides with a vertex of $\pi$, and the weight of the edge otherwise. Let $w_\sigma$ be defined similarly. Then the Fréchet distance of $\pi$ and $\sigma$ with additive edge weights is defined as
    \[\fd^w(\pi,\sigma):= \min\limits_{\text{traversals } (\alpha,\beta)}\;
           \max\limits_{t \in [0,1]}\left( d\bigl( \pi(\alpha(t)),\, \sigma(\beta(t)) \bigr) + w_\pi(\alpha(t)) + w_\sigma(\beta(t))\right)\, .\]
\end{definition}

We first prove that this is indeed results in a lower bound to the Fréchet distance.

\begin{lemma}\label{lem:weightedLowerbound}
    Let $\sigma$ and $\pi$ be two curves, and let $\overline{\sigma}$ and $\overline{\pi}$ be any vertex-restricted simplification of $\sigma$ and $\pi$. Let $w_\sigma$ and $w_\pi$ be non-positive edge weights, such that the edge weight of any edge is at most the negative Fréchet distance of the edge and its corresponding subcurve. Then
    \[\fd^w(\overline{\pi},\overline{\sigma})\leq \fd(\pi,\sigma).\]
\end{lemma}
\begin{proof}
    Consider first a traversal of every edge of $\overline{\pi}$ (resp. $\overline{\sigma}$) with its corresponding subcurve that matches the vertices  of the edge to only the start and end point of the corresponding subcurve. The attained Fréchet distance w.r.t. these joint traversals coincides with the Fréchet distance in the limit. Let $(\alpha_\pi,\beta_\pi)$ (resp. $(\alpha_\sigma,\beta_\sigma)$) be this limit.
    
    Consider a joint traversal $(\alpha,\beta)$ realizing $\fd(\pi,\sigma)$. This traversal induces a joint traversal $(\overline{\alpha},\overline{\beta})$ of $\overline{\pi}$ and $\overline{\sigma}$, where by triangle inequality we have that for all $t\in[0,1]$ it holds that
    \begin{align*}
        \|\pi(\alpha(t)) - \sigma(\beta(t))\|&\geq \|\overline{\pi}(\overline{\alpha}(t)) - \overline{\sigma}(\overline{\beta}(t))\| &- \max_{s:\pi(s)\text{ matched with }\overline{\pi}(\overline{\alpha}(t))\text{ by }(\alpha_\pi,\beta_\pi)}\|\pi(s) - \overline{\pi}(\overline{\alpha}(t))\|\\
        &  &- \max_{s:\sigma(s)\text{ matched with }\overline{\sigma}(\overline{\beta}(t))\text{ by }(\alpha_\sigma,\beta_\sigma)}\|\sigma(s) - \overline{\sigma}(\overline{\beta}(t))\|.\\
    \end{align*}
    But then in particular $\fd(\pi,\sigma)\geq \fd^w(\overline{\pi},\overline{\sigma})$, as the edge weights are always less than these maxima.
\end{proof}

Observe that a traversal realizing the Fréchet distance with edge weights is also a traversal of $\overline{\sigma}$ and $\overline{\pi}$, and hence, together with joint traversals of $\pi$ and $\overline{\pi}$, and $\sigma$ and $\overline{\sigma}$ yield an upper bound to the Fréchet distance of $\pi$ and $\sigma$. Finally, we need to show that the weighted Fréchet distance may be computed with the classical algorithm.

\begin{lemma}
    The Fréchet distance with additive weights can be computed by adding, for each edge $e$ of the simplified curves, negative vertex weights to all eddies on that edges that correspond to the negative weight of the edge $e$. 
\end{lemma}
\begin{proof}
    For this, it suffices to show, that monotonicity events for $(\pi, \sigma)$ are in fact also monotonicity events for these additive edge weights. This is, however, immediately apparent, as the weights are not applied to the vertices of $\overline{\sigma}$ and $\overline{\pi}$, and hence the monotonicity events correspond exactly to the minima of the distance function, as the only weight affecting a column (resp. row) is just a single edge weight, and hence constant. Thus, we can compute the Fréchet distance with additive weights by adding the weights to the eddies, and leaving the rest of the Fréchet distance computing algorithm unchanged.
\end{proof}

\section{Exact arithmetic}
\label{sec:exact}
    In this section we give more details on the exact arithmetic we perform. Ignoring the edge weights for now, every step of the algorithm can classically be performed by comparing square roots of fractions of integers, and fractions of integers, if the input coordinates are integers (or fractions of integers). These comparisons can be done exactly in $O(1)$ time, when provided the fractions inside the roots.

    When introducing edge weights, we need to compare sums of square roots. The number of operations necessary to compare these expressions grows exponentially in the number of square roots in the expressions. To stay within practical realms, we do not do these types of comparisons. By Lemma \ref{lem:weightedLowerbound}, any overestimate of the edge weights also serves as a lower bound. Hence, for these lower bounds we generate an integer fraction that is larger than distance attained.

    \paragraph*{How to resolve square roots.}
    It is common knowledge how to resolve constantly many square roots in an equation in an exact manner.
    The point we wish to make in this appendix section, is that whilst it is publicly known, it is also quite tedious and error-prone.
    Via our reduction of Lemma \ref{lem:weightedLowerbound}, which allows us to use integer over-estimations whenever an expression contains more than two square roots, we only are required to compare expressions of the form:  \[\frac{a}{b}\pm\sqrt{\frac{c}{d}}\leq\frac{e}{f}\pm\sqrt\frac{g}{h}.\]
    These inequalities are  resolved with exact integer arithmetic in the following way:
    
    First, we note that we can relatively easily verify an expression of the form: $\frac{a}{b} \leq \pm\sqrt{\frac{c}{d}}$.
    Indeed, we can first check whether the individual fraction $\frac{a}{b}$ and $\frac{c}{d}$ are positive or negative. 
    This results in a 4-case distinction, where in each case, we can square both sides to evaluate this expression exactly (adding a $-1$ whenever needed). 

    Once expressions of that form can be resolved, we can compare whether:
\[\sqrt{\frac{c}{d}}\leq\frac{e}{f}\pm \sqrt\frac{g}{h}.\]

    Again, we can do case distinction based on whether $\frac{c}{d}$, $\frac{e}{f}$, $\frac{g}{h}$ are positive and negative. For each of the resulting eight cases, there is a way to square both sides such that the inequalities are preserved. 
    The result is an expression with one fewer fraction, on which we recursively apply this strategy. Thus, via a sizeable but constant-depth decision tree, we can evaluate expressions of the form:  \[\frac{a}{b}\pm\sqrt{\frac{c}{d}}\leq\frac{e}{f}\pm\sqrt\frac{g}{h},\]
    in an exact manner using integer arithmetic. 
    We note that for certain input instances, such as the adversarial input derived from Orthogonal Vector instances, evaluating this decision tree is practically costly. Our implementation restricts the input to 32-bit integers (where we use 64-bit integers to resolve these square root equations). 
    We note that it is possible to template the data type to allow for arbitrarily many input bits, but our implementation contains no such templates.

\end{document}